\documentclass[secnumarabic,aps,showpacs,showkey,superscriptaddress,floatfix]{revtex4-1}

\usepackage[dvips]{graphicx}
\usepackage{epstopdf}
\usepackage{amsfonts}
\usepackage{hyperref}
\usepackage{amssymb}
\usepackage{amsmath}
\usepackage{amsthm}
\usepackage{bbm}
\usepackage{subfig}
\usepackage{float}
\usepackage{color}
\usepackage{mdwlist}
\usepackage{listings}
\usepackage[normalem]{ulem}

\newcommand{\Matrix}[4]{\left (\begin{array}{cc} #1 & #2 \\ #3 & #4 \end{array}\right)}

\newcommand{\R}{\mathbb{R}}
\newcommand{\A}{\mathcal{A}}

\renewcommand{\H}{\mathbb{H}}
\newcommand{\1}{\mathbbm{1}}
\renewcommand{\S}{\mathbb{S}}
\newcommand{\E}{\mathcal{E}}
\renewcommand{\P}{\mathbb{P}}

\newcommand{\M}{\mathcal{M}}
\newcommand{\I}{\mathcal{I}}

\newcommand{\avg}[1]{\langle #1 \rangle}

\newtheorem{definition}{Definition}
\newtheorem{proposition}{Proposition}
\newtheorem{thm}{Theorem}
\newtheorem{lemma}{Lemma}
\newtheorem{corollary}{Corollary}
\newtheorem{remark}{\it Remark\/}

\graphicspath{{../figures/}}

\definecolor{orange}{rgb}{1.00,0.50,0.0}

\newcommand{\comment}[1]{{{#1}}}

\begin{document}

\title{The heterogeneous gas with singular interaction:\\ Generalized circular law and heterogeneous renormalized energy}

\author{Luis Carlos Garc\'ia del Molino}\email[]{garciadelmolino@ijm.univ-paris-diderot.fr}
\affiliation{Institut Jacques Monod, CNRS UMR 7592, Universit\'e Paris Diderot, Paris Cit\'e Sorbonne, F-750205, Paris, France}

\author{Khashayar Pakdaman}\email[]{pakdaman@ijm.univ-paris-diderot.fr}
\affiliation{Institut Jacques Monod, CNRS UMR 7592, Universit\'e Paris Diderot, Paris Cit\'e Sorbonne, F-750205, Paris, France}
\affiliation{On leave at Laboratoire de Probabilit\'es et Mod\'elisation Al\'eatoire,  CNRS UMR 7599 Universit\'e Pierre et Marie Curie - Universit\'e Paris Denis Diderot}

\author{Jonathan Touboul}\email[]{jonathan.touboul@college-de-france.fr}
\affiliation{Mathematical Neuroscience Team, CIRB-Coll\`ege de France (CNRS UMR 7241, INSERM U1050, UPMC ED 158, MEMOLIFE PSL*)}
\affiliation{INRIA Paris-Rocquencourt, MYCENAE Laboratory}

\begin{abstract}

We introduce and analyze $d$ dimensional Coulomb gases with random charge distribution and general external confining potential. We show that these gases satisfy a large deviations principle. The analysis of the minima of the rate function (which is the leading term of the energy) reveals that at equilibrium, the particle distribution is a generalized circular law (i.e. with spherical support but non-necessarily uniform distribution). In the classical electrostatic external potential, there are infinitely many minimizers of the rate function. The most likely macroscopic configuration is a disordered distribution in which particles are uniformly distributed (for $d=2$, the circular law), and charges are independent of the positions of the particles. General charge-dependent confining potentials unfold this degenerate situation: in contrast, the particle density is not uniform, and particles spontaneously organize according to their charge. In that picture the classical electrostatic potential appears as a 
transition at 
which order is lost. Sub-leading terms of the energy are derived: we show that these are related to an operator, generalizing the Coulomb renormalized energy, which incorporates the heterogeneous nature of the charges. This heterogeneous renormalized energy informs us about the microscopic arrangements of the particles, which are non-standard, strongly depending on the charges, and include progressive and irregular lattices. 
\end{abstract}


\date{\today}%

\maketitle

%\comment{Indicate either (i) that we assume initial conditions of the form $\mu_0(q,x)=\nu(q)\otimes \rho_0(x)$ or mention that (a) in the case $g(q)$ monotonic, the equilibrium distribution does not depend on the initial condition in the 2d case (not in 1d) and (b) in the case $g(q)$ constant, the stationary distribution depends on the initial condition in the 2d case: $\rho$ depends on $\langle q\rangle$ but not on the specific arrangement of the charges, but the relationship between $q$ and $x$ of the initial condition persist (it is not true that $\mu_{eq}=\nu\otimes \rho_{eq}$ in the case of dependent initial conditions).  }

\setcounter{secnumdepth}{2}
\setcounter{tocdepth}{2}
\bigskip
\hrule	
	\tableofcontents
	
	\medskip
	
	\hrule

\bigskip

\section{Introduction}

Coulomb gases and log-gases have a vast repertoire of applications (e.g.~superconductivity and superfluidity \cite{abrikosov:57,Minnhagen:87}, plasma physics \cite{rogers:73}, string theory \cite{Jokela:2008}, random matrices \cite{Dyson62,Forrester10} and interpolation \cite{Saff:97} to name a few). Many mathematical studies analyze gases of particles with identical charges (henceforth referred to as homogeneous gases) \cite{Rogers93,Benarous97,Benarous98,sandier-serfaty:12,sandier-serfaty:12,rougerie-serfaty:13,chafai2013first}. Heterogeneous multicomponent Coulomb gases consisting of mixtures of positive and negative charges have also been studied as a canonical model of phase transitions (two classical papers are \cite{Chui:76,Amit:80}). 
In stark contrast with these, fewer works have been devoted to heterogeneous gases of non-identical positive charges. Such gases with logarithmic interactions and confined to the circle or the real line were introduced as a means to interpolate between related classical families of random matrices~\cite{Forrester10,Forrester:1989,rider2013solvable,shum2014solvable}. The partition functions of these gases were derived at specific temperatures~\cite{Jokela:2008,Sinclair:2012}. The characterization of the stationary states of general heterogeneous gases with singular repulsive interaction in arbitrary dimensions remains an open problem. In this work, we introduce such a system and describe and analyze its stationary states with an emphasis on the Coulomb gas. Our contributions are twofold. (i) First, we unveil a novel transition such gases undergo, in the ordering of particles with respect to their charges.
% disorder-related transition that is novel 
%in the sense that, on the one hand, it is unlike the one in gases of mixed positive and negative charges and, on the other hand, it is precluded in the one-dimenional heterogenous gases of positively charged particles considered in previous works. 
%with respect to the ones considered in previous works. 
(ii) Second, we derive an expansion of the gas energy to the first two leading orders when the number of particles goes to infinity. This expansion provides precise information about macroscopic and microscopic gas organizations. At the core of the mathematical analyses of these phenomena are the statement and proof of a large deviations principle, together with the relationship of the fluctuations with a renormalized energy for the heterogeneous gas. The scope of these theorems goes beyond the example considered and shall extend to the statistical physics of large heterogeneous particle systems with singular interactions. 

Analogous questions where investigated for homogeneous Coulomb gases. Characterization of equilibrium distributions were undertook in~\cite{frostman}. Large deviations principles with speed $N^2$ of the $N$-particles empirical distribution were demonstrated for 1 dimensional and 2 dimensional homogeneous gases in the context of random matrix theory \cite{Benarous97,Benarous98}, and were recently extended to $d$-dimensional homogeneous gases with general singular repulsive interaction potentials~\cite{chafai2013first}. 

The characterization of the microscopic structure of the stationary distributions entails going beyond the order $N^2$ terms and finding the next leading order terms of the  energy function. A series of papers, \cite{sandier-serfaty:12,sandier-serfaty:12,rougerie-serfaty:13} have expressed these corrections in terms of a renormalized energy inspired from Ginzburg-Landau theory for 1 dimensional, 2 dimensional and $d$-dimensional homogeneous Coulomb gases. One of their main conjectures pertaining to the microscopic structure is that the Abrikosov (regular triangular) lattice is the minimizer of the renormalized energy in 2 dimensions, partially proved in~\cite{sandier-serfaty:12} under the assumption that the minimizer is indeed a lattice. We generalize these results and conjectures to heterogeneous Coulomb gases.
%The nature of the renormalized energy minimizers has been a longstanding issue which remains open. Advances have been done in that direction for short-range interaction~\cite{suto:05,BPT}, methods that do not seem to extend to the type of electrostatic interactions similar to Coulomb. However, it is widely believed that at zero temperature, configurations minimizing the renormalized energy are regular crystals. 

In the remainder of this section, we clarify the above statements by providing an overview of our results. Namely, we first introduce the model and present numerical examples of stationary distributions and their transitions. This is followed by the review of our theoretical contributions. The precise statements of the theoretical results and their proofs are developed in sections \ref{sec:LDP} and \ref{sec:renormalized}. The implications of these results for the gas are presented in section \ref{sec:minimizers}. While our main results are for the Coulomb gas, in the appendix we discuss their extension to disordered gases with more general potentials (appendix \ref{append:GeneralGases}) and on manifolds (appendix \ref{append:manifolds}) and present numerical evidence of the transition phenomenon in such gases. 

\subsection{The gas model}

We consider an $N$-particle gas with Hamiltonian:
\begin{equation}\label{eq:hamiltonian}
H_N=H(q_1,\dots,q_N,x_1,\dots,x_N)=N\sum_{i=1}^N q_ig(q_i) V(x_i)- \sum_{i=1}^N\sum_{j\neq i}q_iq_j W(|x_i-x_j|).
\end{equation}
where $\{x_i\}_{1\leq i\leq N}$ in $\R^d$ with $d\geq 2$ and $\{q_i\}_{1\leq i\leq N}$ denote the positions and charges of the particles. The charges take values sampled from a probability distribution  $\nu\in\mathcal{M}^1(Q)$ with $Q=[q_{\min},q_{\max}]$ and $q_{\min}>0$. $V(x)$ is an external confining potential
%\begin{equation*}
%	\begin{cases}
%		\lim_{\vert x \vert \to \infty} V(x) = +\infty & d\geq 3\\
%		\lim_{\vert x \vert \to \infty} \left(\frac{V(x)}{2} - \log \vert x \vert \right) = +\infty & d=2
%	\end{cases}
%\end{equation*}  
and $W$ is a singular repulsion kernel, i.e. $W'(x) >0$  for $x>0$ and $\lim_{x\to 0^+} W(x) = +\infty$ satisfying additional regularity assumptions (see section~\ref{sec:LDPGeneralGases}). The weight function $g$ allows taking into account general external forces including in particular charge-independent confinements ($g(q)=1/q$). It  plays an important role in our study, and will generally be assumed either constant or a monotonic positive function of $q$. These gases constitute a very rich family, as will be further discussed in detail in section~\ref{append:Univ}.

%The case $g(q)=1$ corresponds to the electrostatic interaction with an external field. 

The classical Coulomb gas corresponds to the special choice of a quadratic confining potential $V(x)=\vert x\vert^2$, constant weight function $g$ and Coulomb interaction kernel which is the Green function of the Laplace operator
\begin{equation}\label{eq:CoulombCondition}
	\Delta W = k_d \delta_0, \qquad i.e. \qquad \begin{cases}
		W(x)=-\frac{1}{\vert x \vert^{d-2}} & d\geq 3\\
		W(x)=\log\vert x \vert & d= 2\\
	\end{cases}
\end{equation}
with $k_d=c_d \vert \S_{d-1}\vert$ and $c_2=1$ and $c_d=d-2$ for $d\geq 3$ ($\vert \S_{d-1}\vert$ is the volume of the unit sphere of dimension $d-1$). In all the text $V$ and $W$ will be assumed to have the forms given above (except otherwise stated).

%Besides the heterogenous Coulomb gas, We also consider interaction kernels  such that \[\nabla W(r) \propto r^{-d+1-\eta}\] 
%for $\eta$ a real parameter. When $\eta>0$, these are the so-called Riesz interactions that have been studied in the frame of equi-repartition of points \cite{}. These have stronger repulsion at small scale than the Coulomb gas and we shall henceforth refer to them as super-Coulomb interactions and gases. Conversely, for similar reasons, we refer interactions with $\eta <0$ as sub-Coulomb.

%In a series of papers, Hardin, Saff and colleagues have examined the large number of particles limit of distributions of weighted Riesz gases in various sets. One major difference between their work and ours is that they consider homogeneous gases in which interparticle interactions are weighted by a function of their respective positions whereas in our case interactions are weighted by the product of charges. 

\subsection{Stationary distributions and transition}

One of the main results of this paper is to show that heterogeneous gases show substantially different phenomenology compared to homogeneous gases. To highlight this point, we present here numerical explorations of the stationary distributions of the two-dimensional heterogeneous Coulomb gas and three choices of $g$ (decreasing, constant and increasing). We have selected these situations as, on the one hand, they are representative of the general case, and on the other hand, they provide the best situation for comparison with homogeneous gases.

%We start with a brief reminder of the stationary distribution of the standard two-dimensional Coulomb gas before going over the three key aspects of the distributions for the heterogeneous gas that we want to highlight here and that will be the motivation and topic of the theoretical investigations of the following sections.

The asymptotic distribution, as well as microscopic arrangements of homogeneous gases reflect the interchangeability of particles within the gas. In the limit of infinite number of particles, the stationary distribution of a homogeneous Coulomb gas in the plane is the uniform distribution on a disc with regular microscopic arrangements at zero temperature (conjectured to be Abrikosov triangular lattices, \cite{sandier-serfaty:12}).

Heterogeneous gases do not have this interchangeability property. Particles with distinct charges experience different confinements and interact with the other particles in a different way. The phenomenology is, in that sense, richer than the homogeneous gases. Figure~\ref{fig:transition} displays numerical simulations of stationary distributions for different two-dimensional gases and will be used as the basis for the description of the phenomena. Panels in the left, center and right correspond respectively to $g$ increasing, constant and decreasing. The colors encode for the charge of the particles, so that  mixed colors indicate disorder whereas separated ones correspond to ``charge ordering''. These panels illustrate three key observations that are listed below and are at the heart of our study.

\begin{figure}[h]
\centering

\includegraphics[width=0.9\textwidth]{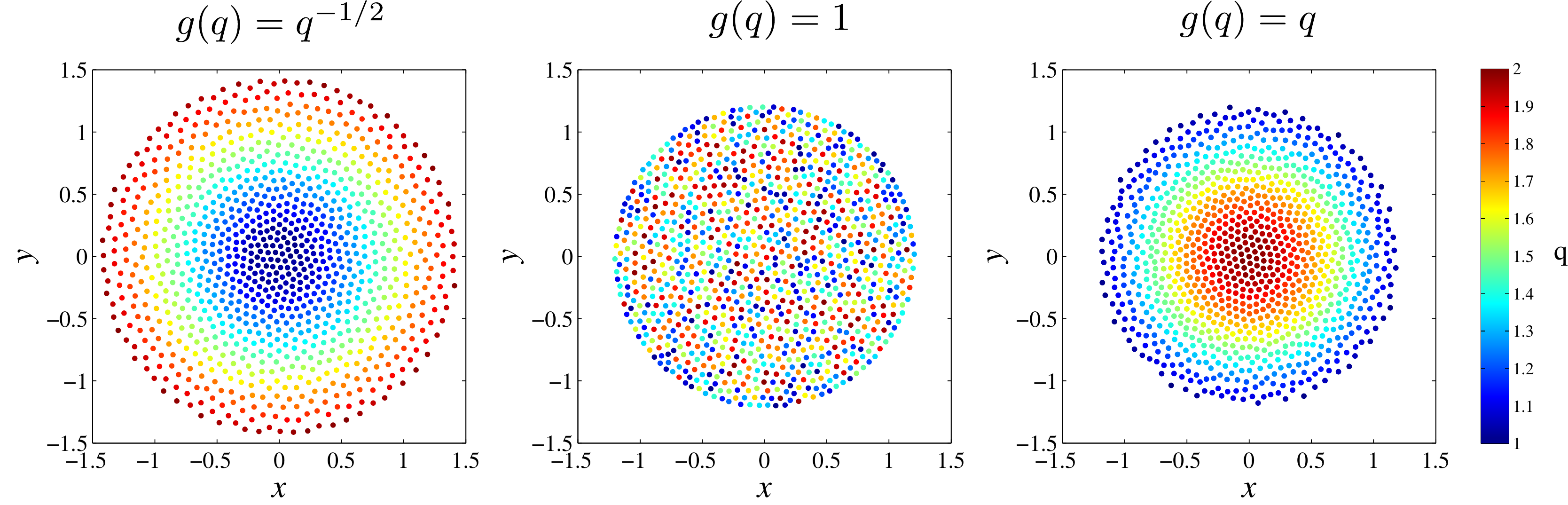}
\caption{Scatter plot of the $N$-particle two-dimensional Coulomb for several choices of $g(q)$. Each point represents a particle's position and the color corresponds to its charge. $N=1000$, $\nu(q)$ uniform in $Q=[1,2]$.}
\label{fig:transition}
\end{figure}

\begin{description}
 \item[Generalized circular law] at the macroscopic level, irrespective of the choice of $g$, the support of the stationary distributions is a disk. However, unlike the standard circular law which further implies that, asymptotically, the particles are uniformly distributed on the disk, in the heterogeneous gas this is not always so. The shapes of radially symmetric distributions for monotonic $g$ tend to be denser at the origin than in the periphery. 
 \item[Charge ordering] at the macroscopic level, the distribution of charges when $g$ is constant is highly disordered. There seems to be no correlation between the values of nearby charges. However, when $g$ is monotonic, we have ``charge ordering''. The system organizes into a highly ordered distribution in which particles are located according to their charge. The signature of this phenomenon is the difference between the rainbow like coloring of the left and right panels of Fig. 1 in contrast with the middle one. Another key observation is that the coloring is reversed between the left and right panels. In other words, the ordering of the charges depends on whether $g$ is increasing or decreasing: while in one case, larger charges take on the outer layers of the gas, in the other case, it is the opposite. The standard Coulomb gas with constant $g$ appears thus as a transition between these two ordered cases.
 \item[Irregular and progressive lattices] at the microscopic level, the particles within gases with $g$ monotonic are organized in a regular structure reminiscent of a triangular Abrikosov structure, whose mesh  size progressively varies as one moves from the center of the disk to the border. The gases at constant $g$ show an irregular lattice because each particle is surrounded by  particles of unrelated charges. 
\end{description}

\begin{remark}
	{Heuristically, the minima of the energy~\eqref{eq:hamiltonian} correspond to stationary states of a system of particles evolving in the potential~\eqref{eq:hamiltonian}. A similar energy and the associated particle system can be defined in one dimension. However in that case particles cannot cross each other, and therefore the stationary state depends on the initial charge distribution. For this reason we will address here the case of gases in dimension $d\geq 2$.}
\end{remark}

\subsection{Overview of theoretical results}

We analyze the phenomena reported in the previous section in the regime where $N\to\infty$. To this end, we shall characterize both the asymptotic \emph{ground state} (macroscopic distribution of the gas) as well as the Gibbs state (microscopic configurations). Both distributions can be obtained by precise derivation of the asymptotic properties of the total energy~\eqref{eq:hamiltonian} in the thermodynamic limit. To this end we introduce the double-layer empirical measure $\hat\mu_N$, the particle density $\rho^N(x)$ (indicating the distribution of the positions of the particles regardless of their charge) and the averaged charge density $\rho_q^N(x)$ at location $x$:
\begin{align}\label{eq:empirical}
\hat\mu_N&=\frac1N\sum_{i=1}^N\delta_{(q_i,x_i)}\ ,\\ 
 \label{eq:Partdens}\rho^N(x)&=\int_Q\hat\mu_N(q,x)\;dq\ ,\\
 \label{eq:Chargedens}\rho_q^N(x)&=\int_Q q\;\hat\mu_N(q,x)\;dq\ .
\end{align}
With these notations we rewrite the $N$-particles energy as $H_N = N^2 h(\hat{\mu}_N)$ where $h$ is an intensive energy independent of $N$ given by:
\[h(\mu)=\int_{Q\times \R^d} qg(q)V(x)d\mu(q,x) - \int\int_{(Q\times\R^d) \times (Q \times\R^d)\setminus D} qq'W(|x-x'|)d\mu(q,x)d\hat\mu(q',x'),\]
where $D=\{(x,x), x\in \R^d\}$. The main mathematical result of the present manuscript is the following expansion for the ground state of $H_N$:
\begin{equation}\label{eq:ClosedForm}
	\boxed{
	\begin{array}{ll}
		\displaystyle{N^2 h(\mu_{\nu}^{\star}) - \frac{1}{2} N \log N \avg{q^2} + N \left(\frac{\alpha_2}{2\pi}\frac{\avg{q}}{q_0} - \frac {\avg{q^2}} 2 \int_{\R^2} \log\left(\frac{\rho_q^{\nu}(x)}{\mu_0\,q_0}\right) \frac{\rho_q^{\nu}(x)}{\mu_0\,q_0}dx\right)} + o(N) \qquad & d=2\\
		\displaystyle{N^2 h(\mu_{\nu}^{\star}) + N^{2-2/d} \frac{\alpha_d}{k_d} \int_{\R^d} \left(\frac{\rho_q^{\nu}(x)}{\mu_0q_0}\right)^{2-2/d} dx} +o(N^{2-2/d}) & d\geq 3\\
	\end{array}
	}
\end{equation}
where $\mu_{\nu}^{\star}$ is a distribution minimizing the intensive energy, $\mu_0$ is the inverse unit length, $q_0$ the unit charge and $\rho_q^{\nu}(x)=\int_{Q} q\mu_{\nu}^{\star}(q,x)dq$ and $\avg{\varphi(q)}=\int_Q \varphi(q) \nu(q)dq$ is the mean charge. The coefficients $\alpha_d$ are universal constants depending on the dimension, that will be explained below. 

The leading term of the energy is derived from a large-deviations principle on the double-layer empirical distribution stated and proved in section \ref{sec:LDP}, together with the existence of the minimizing distribution $\mu_\nu^\star$. This principle further ensures that the empirical macroscopic distribution converges to the set of minima of the rate function. From the computation of $\mu_\nu^\star$ we shall derive the limiting continuous charge distributions and hence provide theoretical grounding for the first and second observations listed above, namely the generalized circular law and the charge ordering and its transition (see section \ref{sec:macro}). 

The sub-leading term is obtained through an expansion of the energy (splitting formula) in terms of a functional called the heterogeneous renormalized energy, and will be rigorously demonstrated in section \ref{sec:renormalized}. In the same way that minimizers of $h$ give the macroscopic organization of the gas, the minimizers of the renormalized energy depict its microscopic structure. From this we derive theoretical support for the third observation listed above, namely the presence of progressive triangular lattices for $g$ monotonic and disordered lattices centered on a regular triangular lattice for $g$ constant (see section \ref{sec:macro}).

\section{Large Deviations Principle for heterogeneous gases}\label{sec:LDP}\label{sec:LDPCoulomb}

%In this section, we present the large deviations principles satisfyed by the 
%double layer empirical measure  $\hat\mu_N(q,x)$ and prove it.

%The characterization of the statistical properties of homogeneous gases has been an active field of research for a long time. First results on existence and uniqueness of solutions, as well as macroscopic distribution of the 1 dimensional homogeneous log-gases were established in \cite{Rogers93} \comment{[check biblio for 1d log-gas]}. Later, large deviations principles with speed $N^2$ of the $N$-particles empirical distribution were demonstrated for 1 dimensional and 2 dimensional gases in the context of random matrix theory \cite{Benarous97,Benarous98}, and 
%were recently extended to the case of $d$-dimensional gases with general singular repulsive interaction potentials \cite{chafai2013first}. 
We now prove the asymptotic properties of heterogeneous gases. We consider a gas at non-zero temperature $T=1/\beta$. The canonical partition function and Boltzmann-Gibbs probability measure on $(\R^d)^N$ with a prescribed empirical charge distribution $\hat{\nu}_N = \frac 1 N \sum_{j=1}^N \delta_{q_j}$ read:
\begin{equation}\label{eq:pdf}
	\begin{cases}
		Z_N^{\hat\nu_N}=\int_{(\R^d)^ N} e^{-\beta H_N}dx_1\cdots dx_N ,\\
		\P_N^{\hat\nu_N}(x_1,\dots,x_N)=\frac{1}{Z_N^{\hat\nu_N}} e^{-\beta H_N}\ .
	\end{cases}
		\end{equation}
We prove that the double layer empirical measure  $\hat\mu_N(q,x)$ satisfies a large-deviations principle. 
%In contrast with the usual context of gases, the system is here disordered: the particles have distinct charges in an interval $Q$, drawn independently from a law $\nu(dq)$. 
This problem is reminiscent of large deviation principles for interacting diffusions with disorder developed in the context of smooth Hamiltonian dynamics in random media~\cite{daipra:96} or for spin glasses~\cite{benarous-guionnet:95}. However these large deviations principles do not apply to singular interaction potentials or infinite time limits. Similar problems for homogeneous systems with singular repulsion were addressed in the context of the eigenvalues 
of random matrices~\cite{Benarous97,Benarous98,hardy2012note,chafai2013first}. We extend these techniques to our case. The main contribution in this proof is precisely to handle the heterogeneity of the charges, and specifically to ensure that the marginal distribution of charges $\int_{\R^d} \mu(x,q\in A)\,dx=\nu(A)$ for all $A\subset Q$ a measurable set.

We denote by $\M_1(Q \times \R^d)$ the set of probability measures on $Q \times \R^d$ and equip the space with the Vasserstein distance. We introduce the rate function restricted to the space $\M_{\nu}$ of probability measures on $\M_1(Q \times \R^d)$ with marginal charge density $\nu$:
\begin{equation}\label{eq:ratefunc}
 \mathcal{I}_{\nu}:\begin{cases}
 	\M_{\nu} &\mapsto \R\\
	\mu &\mapsto \int qg(q) V(x)d\mu(q,x) - \int\int qq' W(|x-x'|)d\mu(q,x)d\mu(q',x') - K_{\nu}
 \end{cases} 
\end{equation}
where $K_{\nu}$ is the minimal value of the intensive energy $h$ on $\M_{\nu}$. %We recall the definition of the intensive energy:
%\[h(\mu)=\frac{1}{N^2}H(\mu)=\int qg(q)V(x)d\mu(q,x) - \int\int qq'W(|x-x'|)d\mu(q,x)d\mu(q',x')\ .\] 
%We show that the double-layer empirical measures satisfies a large deviation principle. 

Let us consider a sequence of double-layer empirical measures $\hat{\mu}_N$, with a marginal charge distribution denoted $\hat{\nu}_N$ (atoms are a sample of size $N$ drawn according to $\nu$). The charge distribution $\hat{\nu}_N$ converges weakly to $\nu$. The sequence $(\hat{\mu}_N)$ is not necessarily contained in $\M_{\nu}$ and therefore, it does not satisfy a large-deviation principle in $\M_{\nu}$. However, for any sequence of empirical measures of charges $\{\hat\nu_N\}_{N\in\mathbb{N}^*}$ in $\M_1(Q)$ such that $\hat\nu_N\rightharpoonup \nu$ when $N\to\infty$, we will show that the set of double-layer empirical measures conditioned on having a marginal charge density $\hat{\nu}_N$:
\[\M_{\hat\nu_N} =\{\mu\in\M_1(Q \times \R^d) : \int \mu(q,x)dx=\hat\nu_N(q)\}\ \]
do satisfy logarithmic upper and lower bounds similar to the classical large-deviation principles with speed $N^2$ for homogeneous gases, allowing to show convergence estimates along arbitrary sequences $\{\hat\nu_N(q)\}_{N\in\mathbb{N}^*}$. These properties imply that the system of heterogeneous particles converges, as $N\to\infty$, towards the minimizers of the rate function computed section~\ref{sec:macro}.

The scheme of the proof is the same in all dimensions and for any type of interaction. The upper bound is proven by direct evaluation of the probability density. The lower bound is slightly more complex. In~\cite{Benarous97}, the authors propose an elegant construction of a particular set of points on the real line, from which they construct a measure whose probability compares to the rate function and lower-bounds the probability we aim at controlling. This construction was generalized in~\cite{Benarous98} where the points now belong to the positive half-plane $\H$ and very recently in \cite{chafai2013first} to points living in $\R^d$ (basis of the proofs proposed in section~\ref{sec:LDPGeneralGases}). In our case of double-layer distributions, one can use these constructions in order to approximate the particle density $\rho$ and then attribute a charge to each particle.

We expose the proof fully in the case of the two-dimensional Coulomb gas and outline the main differences of the extensions to other dimensions or interaction kernels in appendix~\ref{append:GeneralGases}. 

\begin{lemma}\label{lem:GRF}
For Coulomb gases in dimension $2$, the map $\I_{\nu}$ is a good rate function, i.e.:
\renewcommand{\theenumi}{(\roman{enumi})}
\begin{enumerate}
	\item it is lower semi-continuous;
	\item it has compact level sets in $\M_{\nu}$;
	\item $\mathcal{I}_{\nu}$ is convex on $\M_{\nu}$.
\end{enumerate}
This implies that $\I_{\nu}$ is a well-defined good rate function. 
\end{lemma}
\begin{proof} 
Let us first prove some general results, valid in arbitrary dimension, on the intensive energy. We introduce the linear map:
\[\varphi: \begin{cases}
	\M_1(Q\times \R^d) &\mapsto \M(\R^d)^2\\
	\mu &\mapsto \Big(\int_Q q g(q) d\mu(q,x) \;, \; \int_Q q d\mu(q,x) \Big)
\end{cases}.\]
The components of this vector are not probability measures: these are positive measures with total mass $(\int_{Q\times \R^d} q g(q) d\mu(q,x),  \int_{Q\times \R^d} q d\mu(q,x))$. In particular, on $\M_{\nu}$, the image of $\varphi$ is the space $\M_{\langle q g\rangle}(\R^d) \times \M_{\langle q\rangle}(\R^d)$ where $\M_{a}(\R^d)$ is the set of positive measures with mass $a$ and $\langle \chi(q) \rangle =\int_Q \chi(q) d\nu(q)$. The intensive energy $h$ is the combination of the map $\tilde{h}$ and $\varphi$, where $\tilde{h}$ is defined as:
\[\tilde{h}: \begin{cases}
	\M(\R^d)^2 &\mapsto \R\\
	(\rho_{qg}, \rho_q) &\mapsto \int_{\R^d} V(x)d\rho_{qg}(x) - \int_{\R^d\times \R^d\setminus D} W(|x-x'|)d\rho_q(x)d\rho_q(x')
\end{cases}.\]
The first term is linear. The second term is quadratic, and identical to the one involved in large-deviations principles of $d$-dimensional homogeneous Coulomb gases. For $d=2$, it was shown in~\cite{Benarous97,Benarous98} that the map 
\[\rho \in \M_1(\R^2) \mapsto \int_{\R^2} \vert x \vert^2 d\rho(x) - \int_{\R^2\times \R^2\setminus D} W(|x-x'|)d\rho(x)d\rho(x')\]
is a good convex rate function, i.e. lower semi-continuous with compact level sets. Therefore the map
\[\rho_q \in \M_{\langle q\rangle}(\R^2) \mapsto - \int_{\R^2\times \R^2\setminus D} W(|x-x'|)d\rho_q(x)d\rho_q(x'),\]
is convex and $\tilde{h}$ is a good rate function. 

Since the map $\varphi$ is continuous, the intensive energy is therefore lower semi-continuous and lowerbounded. Convexity of $\I_{\nu}$ stems from the convexity of $\tilde{h}$ and the linearity of $\varphi$. Compactness of the level sets arises from the fact the map $\varphi$ is a bounded operator for the Wasserstein distance. 
\end{proof}

\begin{remark}
	The rate function of homogeneous gases is actually strictly convex, and has a unique minimum. Here, this is not necessarily the case. Indeed, the rate function at $\mu\in \M_1(Q\times \R^2)$ only depends on $\varphi(\mu)$, it is therefore constant on the sets 
	\[\M_{\rho_1,\rho_2}=\Big\{\mu\in \M_1(Q\times \R^d) \;;\; \int_Q q g(q) d\mu(q,x)=\rho_1(x) \text{ and } \int_Q q d\mu(q,x)=\rho_2(x) \Big\}.\]
	These may be reduced to a single point: we will show that this is the case for Coulomb gases with strictly monotonic $g$. 
\end{remark}

% \comment{shouldn't we mention the speed $N^2$ in the text ? or is it redundant ?}

\begin{thm}\label{thm:LDP_Coulomb}
For any $\mu \in \M_\nu$, we have:
\begin{equation*}
	\begin{cases}
		\displaystyle{\lim_{\delta \searrow 0} \; \limsup_{N\to\infty}\;\frac 1 {\beta N^2}\log\P[\hat{\mu}_N \in \M_{\hat\nu_N} \cap B(\mu,\delta)] \leq -\mathcal{I}_{\nu}(\mu)}\\
		\ \\
		\displaystyle{\lim_{\delta \searrow 0} \; \liminf_{N\to\infty} \; \frac 1 {\beta N^2}\log\P[\hat{\mu}_N \in \M_{\hat\nu_N} \cap B(\mu,\delta)] \geq -\mathcal{I}_{\nu}(\mu)}\ .
	\end{cases}
\end{equation*}
where $B(\mu,\delta)$ is the L\'evy ball of radius $\delta$ centered at $\mu$. This implies that 
\[\lim_{N\to\infty} d(\hat{\mu}^N,\Lambda^{\star}_{\nu})= 0\]
where $\Lambda^{\star}_{\nu}$ is the set of distributions $\mu_{\nu}^{\star}$ of  $\M_{\nu}$ minimizing $\I_{\nu}$. 
\end{thm}

\begin{proof}
	The first inequality (upper bound) is proved directly by considering the joint eigenvalue density:
	\[
		\P^{\hat\nu_N}_N(B(\mu,\delta)\cap \M_{\hat\nu_N}) \leq \frac{1}{Z^{\hat\nu_N}_N} \int_{(\R^d)^N } dx_1\cdots dx_N\mathbf{1}_{B(\mu,\delta)\cap \M_{\hat\nu_N}} \exp( -\beta N^2 h_{A}(\hat\mu_N))
	\]
	where $h_A(\hat\mu_N) = \min\left(A, h(\hat\mu_N)\right)$, $\hat\mu_N$ is defined as in \eqref{eq:empirical}. We hence have:
	\[
		\P^{\hat\nu_N}_N(B(\mu,\delta)\cap \M_{\hat\nu_N}) \leq \frac{1}{Z^{\hat\nu_N}_N} \int_{(\R^d)^N } dx_1\cdots dx_N\mathbf{1}_{B(\mu,\delta)\cap \M_{\hat\nu_N}}\exp\Big( -\inf_{\hat\mu_N\in B(\mu,\delta)\cap \M_{\hat\nu_N}}\beta N^2h_A(\hat\mu_N) \Big)
	\]
	% Note that our formula slightly differs, in its non-leading terms, from that of~\cite{Benarous98}. 
	% 
	taking the log and dividing by $\beta N^2$ we obtain:
	\[
	 \frac 1 {\beta N^2} \log(\P^{\hat\nu_N}_N(B(\mu,\delta)\cap \M_{\hat\nu_N})) \leq -\frac{\log(Z^{\hat\nu_N}_N)}{\beta N^2}-\inf_{\hat\mu_N\in B(\mu,\delta)\cap \M_{\hat\nu_N}} h_A(\hat\mu_N)  + o(1).
	\]
%	One can now conclude in the same fashion as~\cite[Lemma 2.4]{Benarous98}.	
	Asymptotically, when $N\to \infty$, the first term is bounded by $K_{\nu}:=\min_{\mu\in \M_{\nu}} h(\mu)$. Actually, $ \log(Z_N^{\hat{\nu}_N})/\beta N^2$ converges to $K_{\nu}$, as shown in~\cite{Benarous98} and which can be seen through a saddle-point argument. The right hand side of the inequality still depends on the choice of the sequence $\hat\nu_N$. We note that $h_A$ is continuous for the weak topology (see~\cite[lemma 2.4]{Benarous98}). Therefore, under our assumption that $\hat\nu_N(q) \rightharpoonup \nu$ we can conclude that:
	\[\limsup_{N\to \infty}\frac 1 {\beta N^2} \log(\P^{\hat\nu_N}_N(B(\mu,\delta)\cap \M_{\hat\nu_N})) \leq K_\nu - \inf_{\mu^*\in B(\mu,\delta)\cap \M_{\nu}}h_A(\mu^*)\]
	and using the lower semi-continuity property,
	\[\lim_{\delta \searrow 0} \; \limsup_{N\to\infty}\;\frac 1 {\beta N^2}\log(\P^{\hat\nu_N}_N(B(\mu,\delta)\cap \M_{\hat\nu_N})) \leq  K_\nu - h_A(\mu).\]
The upperbound follows by letting $A$ to infinity and using the monotone convergence theorem. \\

	The particle distribution $\rho$ is approximated by an empirical distribution as was done in the case of homogeneous gases in~\cite[lemma 2.5]{Benarous98}. Indeed, using the same arguments, we may assume that the particle density $\rho$ on $\R^d$ is absolutely continuous with respect to Lebesgue's measure with a bounded and everywhere positive density. This is possible since any distribution $\rho$ can be approximated by such a smooth distribution $\rho_{\varepsilon}$, and the rate function on these approximations $\I_{\nu}(\rho_{\varepsilon})$ converges to $\I_{\nu}(\rho)$, see~\cite[lemma 2.2.]{Benarous98}. We can therefore define a square $\mathcal{C}$ in $\R^2$ containing at least $(1-1/N)$ of the particles. The ensemble $\mathcal{C}$ can be decomposed into $D$ disjoint squares $\{B_l\}_{l\in \{1,\dots, D\}}$ of length proportional to $1/\sqrt{N}$, in each of which are placed a number of points $a_l$ based on the density $\rho$ in each of these squares. Such a 
construction yields an empirical measure whose mass is smaller that one. This is completed by adding points outside $\mathcal{C}$. We therefore constructed a set of $N$ points  $(\tilde{X}_1,\cdots, \tilde{X}_N) \in (\R^2)^N$, which moreover satisfy the important properties:
	\renewcommand{\theenumi}{(\roman{enumi})}
	\begin{enumerate}
		\item the number of points $\tilde{X}_j$ is related to the mass contained in the square they are contained in,
		\item the distance to the boundary of the square, as well as distances between two points, are lower-bounded by $C/\sqrt{N}$ for some constant $C$.
	\end{enumerate}
The thus constructed empirical distribution is close from $\rho$, in the sense that the distance (in total variation) between these two measures is arbitrarily small as soon as sufficiently fine partitions of $\mathcal{C}$ are considered. 
	
	Based on this construction, we derive a set of double-layer empirical measure by attributing charges according to $\hat{\nu}_N$. There are at most $N!$ possible distributions corresponding to the charges. 
	
	For sufficiently small $\varepsilon>0$ we can define small non-overlapping balls centered at $\tilde{X}_i$ with radius $\varepsilon/N$. The union of these balls is denoted $D^{\varepsilon}$. From this construction, following exactly the same algebra as in \cite{Benarous98}, we obtain: \\
	\begin{align*}
		\P^{\hat\nu_N}_N(B(\mu,\delta)\cap \M_{\hat\nu_N})& \geq \sum_{\sigma\in \mathcal{S}_N}\frac{B_N}{Z_N^{\hat\nu_N}} \exp\Bigg(-\beta N \sum_{i=1}^{N} q_{\sigma(i)}g(q_{\sigma(i)}) V(\tilde{X}_i) \\
		& \qquad +\beta\sum_{i=1}^N\sum_{j\neq i}q_{\sigma(i)}q_{\sigma(j)} W(\vert \tilde{X}_i - \tilde{X}_j \vert) - \beta\log(1-\frac{2\varepsilon}{C})\Bigg)
	\end{align*}
	where $C$ denotes a constant independent of $N$ (possibly depending on the parameter $\varepsilon$), $B_N$ a constant depending on all parameters such that $ \log(B_N)=o(N)$. We noted $\mathcal{S}_N$ the set of permutations of $\{1\cdots N\}$, which is of size $N!$. We therefore have:
	\begin{equation*}
		\P^{\hat\nu_N}_N(B(\mu,\delta)\cap \M_{\hat\nu_N})\geq N! \frac{B_N}{Z_N^{\hat\nu_N}} \exp\big(-\beta N^2 h(\mu)+ R(\varepsilon,N)\big),
	\end{equation*}
	 where $R(\varepsilon,N)$ tends to zero when $\varepsilon\to 0$, and is negligible compared to $N$ as $N\to\infty$. This term can be obtained explicitly using the same techniques as in~\cite{Benarous98}, and bounded taking into account the fact that the charges belong to a bounded interval $Q$. From this expression, one can now see that:
	\[
		\frac 1 {\beta N^2} \log(\P^{\hat\nu_N}_N(B(\mu,\delta)\cap \M_{\hat\nu_N})) \geq -\frac {\log(Z_N^{\hat\nu_N})} {\beta N^2} - h(\mu) + \frac{\log(N!)}{\beta N^2}+ o(1)
	\]
	and therefore in the limit $N\to \infty$, then $\varepsilon \to 0$, we conclude that:
	\[
		\lim_{\delta \searrow 0}\liminf_{N\to \infty} \frac 1 {\beta N^2}\log(\P^{\hat\nu_N}_N(B(\mu,\delta)\cap \M_{\hat\nu_N})) \geq  K_\nu - h(\mu)	
	\]
	which ends the proof of upper and lower bounds. 
	
	In particular these inequalities imply that for any $\mathcal{A}\subset \M_{\nu}$, 
	\begin{multline*}
		-\inf_{\mu\in\mathring{\mathcal{A}}} \I_{\nu}(\mu)\leq \lim_{\delta \searrow 0} \lim\inf_{N\to\infty} \frac 1 {\beta N^2}\log\P[\hat{\mu}_N \in \M_{\hat\nu_N} \cap \mathcal{A}_{\delta}]\\
		\leq \lim_{\delta\searrow 0}\lim\sup_{N\to\infty} \frac 1 {\beta N^2}\log\P[\hat{\mu}_N \in \M_{\hat\nu_N} \cap \mathcal{A}_{\delta}]\leq -\inf_{\mu\in\bar{\mathcal{A}}} \I_{\nu}(\mu)
	\end{multline*}
	where $\mathcal{A}_{\delta} = \{\mu\in \M_1(Q\times\R^d), d(\mu, \mathcal{A})\leq \delta\}$ where $d$ is the Wasserstein distance, and $\mathring{\mathcal{A}}$ (respectively $\bar{\mathcal{A}}$) denote the interior (resp, the closure) of $\mathcal{A}$ in $\M_{\nu}$. 
	The convergence result is a classical consequence of these bounds together with the Borel Cantelli lemma.
\end{proof}

\begin{remark}
	If there is uniqueness of the minimizer, then the system is self-averaging, i.e. that for any sequence $\hat{\nu}_N$, the equilibrium distribution of the gas is identical and is the unique minimizer of $\I_{\nu}$. This means that we have a quenched convergence theorem, i.e. that for almost all realization of the charges, the system converges towards the same distribution.
\end{remark}

\section{The heterogeneous renormalized energy}\label{sec:renormalized}

\comment{The large deviation principle provides the leading term of the energy, which is of order $N^2$. Minimization of the leading term of the energy generally yields macroscopic distributions, and large particle systems sample these distributions in densely packed particle ensembles. There are several ways to organize microscopically so that, in the thermodynamic limit, one obtains a given macroscopic distribution. In order to characterize the microscopic properties, the method proposed in~\cite{sandier-serfaty:12,sandier-serfaty:13,rougerie-serfaty:13} consists in computing next to leading order terms in the $N$-particles energy. These are related to the renormalized energy through the so-called \emph{splitting formula} and take into account the effect of local and self interactions. These interactions are due to the fact that the particles are actualy punctual which has no impact on the macroscopic scale. Equilibrium microscopic configurations at vanishing temperature are minimizers of these additional 
terms \footnote{These configurations do not exactly minimize the renormalized energy, but except with exponentially small probability, the averaged renormalized energy converges to the minimum of $\E$ at vanishing temperature.}.  
% A conjecture, supported by mathematical evidences~\cite{sandier-serfaty:12}, tends to indicate that in two dimensions, the minimum of the renormalized energy is achieved on the regular triangular lattice.
%For homogeneous Coulomb gases, the energy of a $N$-charges system with distribution $\mu$ can be written as:
%\begin{equation}\label{eq:HomogeneousFormulae}
%	\begin{cases}
%		H_N=H(\mu) + \frac n 2 \log(n) + n \left[\frac{\alpha_2}{2\pi} - \frac 1 2 \int \mu(x) \log(\mu(x))dx\right] + o(n)& d=2\\
%		H_N=H(\mu) + \frac{n^{2-\frac 2 d}\alpha_d}{k_d} \int \mu(x)^{2-\frac 2 d}(x) dx + o(n^{2-\frac 2 d}) & d\geq 3\\
%	\end{cases}
%\end{equation}
%where $\alpha_d$ can be evaluated. 

One fundamental distinction between homogeneous and heterogeneous gases is that heterogeneity breaks the particle exchange symmetry so that the location of the particles is, in the very definition of the energy, entangled with the value of their charge. We will here introduce the definition of the heterogeneous renormalized energy in arbitrary dimension and show the splitting formula relating the $N$-particles energy to the heterogeneous renormalized energy functional. We distinguish two situations: (i) the case where $g$ is monotonic, in which case we have a unique minimum of the rate function as shown in section~\ref{sec:macro}, and (ii) the case where $g$ is constant, in which case we have infinitely many minima of the rate function, all satisfying the fact that $\rho_q$ is a circular law. As in previous sections, these results apply to the Coulomb gas. Generalizations are outlined in appendix \ref{append:GeneralGases}.}

\subsection{Strictly monotonic $g$}\label{sec:RenormalizedEnergy}
\comment{We start by defining the $N$-particles heterogeneous renormalized energy for systems with finite $N$ and later we will show its relation to the $N$ order terms of the energy through the splitting formula.
\begin{definition}[$N$-particles Heterogeneous Renormalized Energy]\label{def:RenormalizedEnergy}
	Let $m\in\M_1(Q\times \R^d)$ be a double-layer probability measure such that $\int f(q)dm(q,x)\in C^0(\R^d)$ for all continuous function $f$. 
	Let $E$ be the vector field given by
	\[\nabla \cdot E (x) = \int_Q q k_d (d\mu_{\Lambda}(q,x)- dm(q,x)) \qquad \nabla \times E = 0\]
	with $\mu_{\Lambda}=\sum_{(q,x)\in\Lambda} \delta_{(q,x)}$ and $\Lambda$ is a discrete set of points of $Q\times \R^d$.
	
	For $\chi:\R^d\mapsto \R$ a continuous function, the \emph{$N$-particles heterogeneous renormalized energy} of the vector field $E$ is defined as:
	\[\E(E,\chi)=\lim_{\eta\to 0} \frac 1 2\left( \int_{\R^d\setminus\cup_{(q,x)\in \Lambda} B(x,\eta)} \chi(y)\vert E(y)\vert^2\,dy - k_d W(\eta)\sum_{(q,x)\in \Lambda}\chi(x)q^2\right)\ .\]
\end{definition}
When the system is homogeneous and all particles have charge $q=1$, this functional corresponds to the $N$-particles renormalized energy in \cite{sandier-serfaty:12,sandier-serfaty:13}.} \\ 
\begin{proposition}
This quantity is indeed well defined. 
\end{proposition}
\begin{proof}
	Thanks to the properties of the Coulomb gas and Laplace operators, for $(q,x)\in \Lambda$, we can write the vector field $E$ in the neighborhood of $x$ as $E(y) = q(\nabla_y W(y-x) + f(y))$ where $f$ has no singularity at $x$ (it is $C^1$ at  $x$), therefore the added term precisely compensates the divergence of the integral. The definition of the limit follows (see also the proof of theorem~\ref{thm:splitting}).  
\end{proof}

For heterogeneous gases, in the case where there exists a unique minimizer of the rate function (or when the system converges towards a uniquely defined measure $\mu_{\nu}^{\star}$), the $N$-particles energy satisfies the following splitting formula:
\begin{thm}[Splitting formula]\label{thm:splitting}
	The $N$-particles energy satisfies the following decomposition, for $\tilde{\mu}_N=N\hat{\mu}_N$.
	\begin{equation}\label{eq:split}
		H_N=H(\tilde{\mu}_N)=N^2\I(\mu^{\star}_{\nu})+ 2N \int \zeta(q,x) d\tilde{\mu}_N(q,x) + \frac 1 {k_d}\E(\nabla E_N, \mathbbm{1}_{\R^d}) 
	\end{equation}
	where 
	\begin{equation}\label{eq:zeta}
		\zeta (q,x)=\frac{g(q)}{2} V(x) + q \int_{Q\times \R^d} q' W(\vert x-x'\vert)d\mu_{\nu}^{\star}(q',x')
	\end{equation}
	and
	\[E_N(x)=-k_d \int_{Q} q \Delta_x^{-1}(\tilde{\mu}_N-N\mu_{\nu}^{\star}) = \int_{Q\times \R^d} q W(\vert x-x'\vert) d(\tilde{\mu}_N-N\mu_{\nu}^{\star})(q,x') .\]
\end{thm}

\begin{proof}
	This theorem is proved by thoroughly evaluating the $N$ particles energy around its limit $\mu_{\nu}^{\star}$. Let us denote $D\subset \R^d\times \R^d$ the diagonal $\{(x,x) ; x \in \R^d\}$. We have:
	\begin{align*}
		H_N:=H(q_1,\cdots, q_N,x_1,\cdots, x_N)&= N \int_{Q\times\R^d} q g(q) V(x)d\tilde{\mu}_N(q,x) \\
		&\quad -\int_{Q\times\R^d}\int_{Q\times\R^d} qq' W(\vert x-x'\vert) d\tilde{\mu}_N(q,x) d\tilde{\mu}_N(q',x')
	\end{align*}
	and by defining $\delta_N=(\tilde{\mu}_N-N\mu_{\nu}^{\star})$, we obtain:
	\begin{equation}\label{eq:EnergyDev}
		H_N=N^2 \I(\mu_{\nu}^{\star}) + 2N\int_{Q\times\R^d} \zeta(q,x) d\delta_N(q,x)-\int_{Q\times Q \times D^c} qq' W(\vert x-x'\vert) d\delta_N(q,x) d\delta_N(q',x')
	\end{equation}
	with $\zeta$ given by equation~\eqref{eq:zeta}. We remark that:
	\[\begin{cases}
		\zeta(q,x) = 0  & x\in \textrm{supp}(\mu_{\nu}^{\star})\\
		\Delta \zeta(q,x) =\frac {g(q)} 2 \Delta V (x) & x\notin \textrm{supp}(\mu_{\nu}^{\star})\\	
	\end{cases}.
	\]
	The vanishing of $\zeta$ within the support in $\R^d$ of $\mu_{\nu}^{\star}$ stems from the evaluation of the forces at equilibrium performed in section~\ref{sec:macro}. The second term thus corresponds precisely to that of the splitting formula~\eqref{eq:split}. It only remains to show that the last term can be written in terms of the renormalized energy. 
	
	First of all, since $\tilde{\mu}_N$ and $N\mu_{\nu}^{\star}$ have the same mass and compact support, $E_N(x)=O(1/\vert x \vert)$ and $\nabla E_N(x)=O(1/\vert x \vert^2)$ at infinity. Therefore we first approximate the integral term in definition~\ref{def:RenormalizedEnergy} by an integral over a finite domain and use the Green formula:
	\begin{multline*}
		\int_{B(0,R)\setminus\cup_{i=1}^N B(x_i,\eta)} \vert \nabla E_N(x)\vert^2 dx = \int_{\partial B_R} E_N(x)\nabla E_N(x)\cdot e + \sum_{i=1}^N \int_{\partial B(x_i,\eta)} E_N \nabla E_N\cdot e \\
		- \int_{B(0,R)\setminus\cup_{i=1}^N B(x_i,\eta)} E_N(x) \Delta E_N(x)\,dx
	\end{multline*}
	where $e$ generically denotes the outer unit normal vector of the surface considered. The first term vanishes when $R\to\infty$ from the decay properties of $E_N$. The last term is precisely equal to 
	\[-k_d\int_{Q}\int_{B(0,R)\setminus\cup_{i=1}^N B(x_i,\eta)} q E_N(x) d\delta_N(q,x) = N k_d\int_{Q}\int_{B(0,R)\setminus\cup_{i=1}^N B(x_i,\eta)} q E_N(x) d\mu_{\nu}^{\star}(q,x) \]
	because the integral is taken outside of the support of $\tilde{\mu}_N$. The integral on the contour of the balls around the singularities can be computed by splitting the function $E_N(x)$ into a singular and a regular part in the neighborhood of each particle. In detail, we define the vector field $E_N^i(x) = E_N(x) - q_i W(\vert x-x_i\vert)$ which is $C^1$ at $x_i$, hence, as $\eta\to 0$, we have,
	\[\int_{\partial B(x_i,\eta)} E_N(x)\nabla E_N \cdot e = k_d q_i(q_i W(\eta) +  E^i_N(x_i)) + o(1),\]
	and hence,
	\begin{multline}\label{eq:IntBR}
		\int_{B(0,R)\setminus\cup_{i=1}^N B(x_i,\eta)} \vert \nabla E_N(x)\vert^2 dx = k_d W(\eta) \sum_{i=1}^N q_i^2 + k_d \sum_{i=1}^N q_i E_N^i(x_i) \\
		-N k_d\int_{Q}\int_{B(0,R)\setminus\cup_{i=1}^N B(x_i,\eta)} q E_N(x) d\mu_{\nu}^{\star}(q,x) + o(1)\ .
	\end{multline}
	Letting $R\to\infty$ and $\eta\to 0$, we have, by definition:
	\begin{equation}\label{eq:CalculatedRE}
		\E(\nabla E_N,\mathbbm{1}_{\R^2})=k_d \sum_{i=1}^N q_i E_N^i(x_i) -N k_d\int_{Q}\int_{\R^d} q E_N(x) d\mu_{\nu}^{\star}(q,x)\ .
	\end{equation}
	Moreover, since $E_N^i(x)=\int_{Q\times \R^d} q W(\vert x-x'\vert) d[\delta_N(x) - \delta_{(q_i,x_i)}]$, we have:
	\[ \int_{Q\times \R^d\setminus\{x_i\}} q W(\vert x_i-x\vert) d\delta_N(q,x)= E_N^i(x_i),\]
	and at non-singular points for $\tilde{\mu}_N$:
	\[\int_{{Q\times \R^d\setminus\{x\}}} q W(\vert x_i-x\vert) d\delta_N(q,x)=E_N(x).\]
	All together, we have:
	\[\int_{Q\times Q \times D^c} qq'W(\vert x-y \vert )d\delta_N(q,x)d\delta_N(q',y) = \sum_{i=1}^N q_i E_N^i(x_i) - N \int_{Q\times\R^d}q  E_N(x)d\mu_{\nu}^{\star}(q,x), \]
	which, together with equation~\eqref{eq:CalculatedRE}, finishes the proof. 
\end{proof}

\comment{Here, in order to deal with the singularities in $W$ we truncate the interaction potential inside the balls of radius $\eta$ centered on the position of the particles similarly to \cite{sandier-serfaty:12,sandier-serfaty:13}. In \cite{rougerie-serfaty:13} an alternative definition of the $N$-particles renormailzed energy is used where instead of truncating the interaction potential they use smeared charges. In that case the proof of the splitting formula is slightly different but it could also be adapted for heterogeneous gases. In \cite{rougerie-serfaty:13} the authors analyze the advantages and disadvantages of using either definition.\\

In order to evaluate the impact of local interactions on these energy terms we have to blow up the system, i.e.\ rescale the particles position to a scale where they are not densely packed. Defining $x_i'=N^{1/d} x_i$ the rescaled version of \eqref{eq:split} yields formulae similar to those of the homogeneous gas. }

\begin{corollary}[Scaling properties]\label{cor:SplitFormulaRescaled}
	Defining $E_N'(x)=\int_{Q\times \R^d} q W(\vert x-x'\vert )\left(\sum_{i=1}^N \delta_{(q_i,x_i')} -\mu_{\nu}^{\star} (q,N^{-1/d} x')\right)$,
	we have:
	\begin{equation}\label{eq:splitScaling}
		H_N=\begin{cases}
			\displaystyle{N^2\I(\mu^{\star}_{\nu})+ 2N \int \zeta(q,x) d\tilde{\mu}_N(q,x) + \frac 1 {k_d}\E(\nabla E_N', \mathbbm{1}_{\R^d}) -\frac {\log N}{2} \sum_{i=1}^N q_i ^2} & d=2\\
			\\
			\displaystyle{N^2\I(\mu^{\star}_{\nu})+ 2N \int \zeta(q,x) d\tilde{\mu}_N(q,x) + \frac {N^{1-2/d}} {k_d}\E(\nabla E_N', \mathbbm{1}_{\R^d})} & d\geq 3\\
		\end{cases}
	\end{equation}
\end{corollary}
\begin{proof}
	Indeed, by the change of variable formula, we have:
	\[\int_{\R^d\setminus \cup_{i=1}^N B(x_i,\eta)}\vert \nabla E_N(x)\vert^2 dx=N^{1-2/d} \int_{\R^d\setminus \cup_{i=1}^N B(x_i',N^{1/d}\eta)} \vert \nabla E_N'(x')\vert^2 dx'.\]
	Moreover, by definition, we have:
	\[W(\eta)=\begin{cases}
		W(\eta \sqrt{N}) + \frac 1 2 \log(N) & d=2\\
		\\
		N^{1-2/d} W(\eta N^{1/d}) & d\geq 3
	\end{cases}\]
	Therefore, by definition of the renormalized energy, we obtain:
	\[\E(\vert \nabla E_N\vert, \mathbbm{1}_{\R^d}) = \begin{cases}
		\displaystyle {\E(\vert \nabla E_N'\vert, \mathbbm{1}_{\R^d}) - \frac 1 2 \log N \sum_{i=1}^N q_i^2 } & d=2\\
		\\
		\displaystyle {N^{1-2/d} \E(\vert \nabla E_N'\vert, \mathbbm{1}_{\R^d})} & d\geq 3.
	\end{cases}\]
	Using the result of theorem~\ref{thm:splitting} ends the proof.
\end{proof}

\comment{These results provide first-order corrections to the $N$-particles energy. The evaluation of the $N\to\infty$ limit of this quantity presents several difficulties and has to be carried out carefully.
%both stemming from the fact that in the limit the particles are spread all over $\R^d$ and can be arbitrarily far away from any blow up position $x_i$. 
The first problem one encounters is that the integral \eqref{eq:IntBR} has a singularity when $R=|x_i|$. Since in the $N\to\infty$ limit particles are spread all over $\R^d$ the $R\to\infty$ limit of this integral is not well defined. To deal with this problem one uses a continuous cutoff function $\chi_R$ with support in $B_R$ and equal to $1$ on $B_{R-1}$ \footnote{In~\cite{sandier-serfaty:12,sandier-serfaty:13} the authors use cubic domains $[-R,R]^ d$ instead of balls, but they show that the value of the limit does not depend on the shape of the cutoff function.}.

Furthermore, $\E$ scales as $N$, so in order to have a bounded quantity we use an energy density. In detail, we now consider measures $\mu$ such that 
\begin{equation}\label{eq:Am}
	\begin{cases}
		\mu=\sum_{(q,x)\in \Lambda} \delta_{(q,x)} \text{where $\Lambda$ is a discrete (possibly infinite) subset of $Q\times \R^d$}\\
		\frac{\mu(Q\times B_R)}{\vert B_R\vert} \qquad \text{is bounded by a constant independent of } R>1.
	\end{cases}
\end{equation}
Under this assumption, the renormalized energy $\E(\mu,{\chi}_{R})$ diverges as $\vert B_R\vert$. One may then define an \emph{Averaged Renormalized Energy} as in~\cite{sandier-serfaty:12,sandier-serfaty:13,rougerie-serfaty:13}:
\begin{definition}[Averaged Heterogeneous Renormalized Energy]
\[\E_{\infty} (E) = \limsup_{R\to\infty} \frac{\E(E, {\chi}_{R})}{\vert B_R\vert}.\] 
\end{definition}

%The second problem one encounters is that, since the particles are now scattered all over $\R^d$, there is no preferential point to center the blow up and the  limit might depend on this centering point. The approach proposed in~\cite{sandier-serfaty:12,sandier-serfaty:13,rougerie-serfaty:13} is to take an average over all possible centering points. By means of an ergodic theorem the authors are able to show that this average is equal to the integral centered around any of the points. Their ergodic theorem is very general and it can also be applied to the heterogeneous renormalized energy.

}

\comment{
% 
% We now formulate the following theorem:
% \begin{thm}\label{thm:minRenorm}
Let us now fix $m(x)$ a continuous density function in $\R^d$ and define $\A_{m(x)}$ the set of vector fields $E$ such that: 
\[\nabla \cdot E (x) = k_d \Big (\int_Q q  \mu(q,x) dq - m(x) \, q_0\mu_0 \Big) , \qquad \nabla \times E = 0\]
with $\mu$ a distribution of charge satisfying hypotheses~\eqref{eq:Am}. Similarly to what is proved for homogeneous Coulomb gases in~\cite{sandier-serfaty:12,rougerie-serfaty:13}, the next-to-leading order term in the $N$-particles heterogeneous gas energy is equivalent to $N$ times the averaged heterogeneous renormalized energy, and its minimum on the set of vector fields $\A_{m(x)}$ can be expressed as
\begin{equation}\label{eq:min}
	\min_{\A_m(x)}\E_{\infty}(\vert \nabla E\vert) = \begin{cases}
		m(x)\; \alpha_2 + \langle q^2 \rangle_{\mu} \frac {m(x)}{ 2} \log m(x)  & d=2\\
		m(x)^{2-2/d} \alpha_d & d\geq 3.
	\end{cases}
\end{equation}
with $\alpha_d$ the minimum of the renormalized energy density on $\A_1$ (we recall that $q_0>0$ the charge unit and $\mu_0$ the unit density). 
% \end{thm}

It is not in the scope of this work to give a formal proof of this statement that shall follow the steps of the general theory developed for homogeneous gases. We nevertheless outline the main ideas behind it, and the interested reader is referred to~\cite{sandier-serfaty:13,rougerie-serfaty:13} for more details. In order to demonstrate these properties, one can start by noting that the scaling property used in the proof of corollary~\ref{cor:SplitFormulaRescaled} allows showing that, if $E\in\A_{m}$ with $m>0$ constant, then the vector field $E'(x)=\frac{1}{m^{1-1/d}}E(\frac{x}{m^{1/d}})$ (belonging to $\A_{1}$) satisfies the relationship:
\[
	\E_{\infty}(\vert \nabla E\vert) = \begin{cases}
		\displaystyle {m \;\E_{\infty}(\vert \nabla E'\vert) + \langle q^2 \rangle_{\mu} \frac {m} 2 \log (m) } & d=2\\
		\displaystyle {m^{2-2/d} \E_{\infty}(\vert \nabla E'\vert)} & d\geq 3.
	\end{cases}
\]
Thanks to this one-to-one correspondence between elements of  $\A_m$ and elements of $\A_1$, if a minimum of the averaged renormalized energy on $\A_1$ exists, then a minimum exists on all $\A_m$ and relationship~\eqref{eq:min} is valid. Proving these results therefore amounts to showing that (i) the next-to-leading order term in expansion of the energy is $N$ times the averaged heterogeneous renormalized energy, (ii) there exists a minimum of the renormalized energy on $\A_1$ and (iii) the latter scaling relationship extends to non uniform measures. A rigorous proof would require upperbounding and lowerbounding the $N\to\infty$ limit of the $N$-particles renormailzed energy. This is performed by introducing a group of translations that involve the averaged renormalized energy and applying specific results of ergodic theory proved in \cite{sandier-serfaty:13}. Physically, (iii) relies on the fact that the renormalized energy is a very local quantity depending only on the environment of the charges at a 
distance of order $N^{-1/d}$. The result follows from averaging over the possible charges in a small open set around location $x\in \R^d$, at which scale the averaged background charge density is essentially uniform equal to $\rho_q(x)=\int_Q \mu_{\nu}^{\star}(x,q) dq$. }

Let us emphasize that the limit expression is understood in an annealed sense: it depends on the average charge at a specific location. Therefore, for one given realization of the gas, it may not directly relate to the microscopic arrangements. This will however be the case when the value of the charge at equilibrium at a specific location is deterministic. Besides the trivial case of homogeneous gases, we will demonstrate in section~\ref{sec:macro} that this is the case for $g$ monotonic. In that situation, the particles at equilibrium arrange deterministically according to their charges, and the equilibrium distribution splits as $\mu_{\nu}^{\star}(q,x) = \delta_{q(x)}\rho(x)$. Therefore, microscopic arrangements minimize the renormalized energy~\eqref{eq:ClosedForm}.
% , which can be further simplified using the fact that $\rho_q(x)=q(x)\rho(x)$:
% \begin{equation}
% 	H_N=\begin{cases}
% 		\displaystyle{N^2 h(\mu_{\nu}^{\star}) - \frac{1}{2}\log N\sum_{i=1}^N q_i^2 	 + N \left(\alpha_2\frac{\avg{q}}{2\pi q_0} - \frac {\avg{q^2}} 2 \int_{\R^2} \log\left(\frac{q(x)\rho(x)}{q_0\mu_0}\right) \frac{q(x) \rho(x)}{q_0\mu_0} dx\right)} + o(N) & d=2\\
% 		\displaystyle{N^2 h(\mu_{\nu}^{\star}) + N^{2-2/d} \frac{\alpha_d}{k_d} \int_{\R^d} \left(\frac{q(x)\rho(x)}{q_0\mu_0}\right)^{2-2/d} dx} +o(N^{2-2/d}) & d\geq 3.\\
% 	\end{cases}
% \end{equation}

\subsection{Constant $g$}\label{sec:RenormConstantg}
When $g$ is constant, the energy can be written only in terms of the charge density $\tilde{\rho}_q^N (x)= N\rho_q^N (x)=\sum_{j=1}^N q_j \delta_{x_j}(x)$:
\begin{align*}
	H_N = N \int_{\R^d} V(x) d\tilde{\rho}_q^N (x) -\int\int_{\R^d\times\R^d\setminus D} W(\vert x-x'\vert) d\tilde{\rho}_q^N (x)d\tilde{\rho}_q^N (x')
\end{align*}
Moreover, we note (see also the remark in the proof of lemma~\ref{lem:GRF}) that the good rate function in the situation $g$ constant only depends on $\rho_q$. We will demonstrate in section~\ref{sec:macro} that the minimization of this energy yields an equilibrium distribution $\rho_q$ which is a circular law. This scenario can be seen, in that view, as a particular case of the work developed for homogeneous gases. The works of~\cite{sandier-serfaty:12,rougerie-serfaty:13} ensure that the ground state of the $N$ particles energy enjoys the expansion:
\begin{equation}\label{eq:ClosedFormConstantG}
{
	\begin{array}{ll}
		\displaystyle{N^2 h(\rho_q^{\nu}) - \frac{\langle q^2 \rangle}{2} N\log N + N \left(\frac{\alpha_2}{2\pi}\langle q \rangle - \frac 1 2 \int_{\R^2} \log\left(\rho_q^{\nu}(x)\right) \rho_q^{\nu}(x)dx\right)} + o(N) \qquad & d=2\\
		\displaystyle{N^2 h(\rho_q^{\nu}) + N^{2-2/d}\frac{\alpha_d}{k_d} \int_{\R^d} (\rho_q^{\nu}(x))^{2-2/d} dx} +o(N^{2-2/d}) & d\geq 3\\
	\end{array}
	}
\end{equation}
where we denoted $h(\rho_q^{\nu})$ the common value of the energy for any double-layer distribution with charge distribution $\rho_q^{\nu}$. This is again, exactly formula~\eqref{eq:ClosedForm}. But in that case, the averaging over charges involved in the calculation of the renormalized energy is non-trivial, since, as we will see in section~\ref{sec:macro}, at equilibrium, there is no correlation between the charge and the position of a particle. The averaged result does not directly provide information on the microscopic arrangements for a given realization (see section~\ref{sec:Constantg}).

\section{Energy minimizers}\label{sec:minimizers}\label{sec:macro}
%%%%%%%%%%%%%%
%\sout{In this section, we use our theoretical results to characterize the equilibrium distributions of the heterogeneous gases. We start by investigating the macroscopic distributions, that are given by the minimizer of the rate functions and reveal the ordering transition, before characterizing the microscopic distributions through the minimization of the renormalized energy. }

After establishing the mathematical results of the paper in sections \ref{sec:LDP} and \ref{sec:renormalized}, in this section we show how the properties of the stationary distribution of the heterogeneous gas are derived from the evaluation of the terms in the expansion of the energy \eqref{eq:ClosedForm}.

\subsection{Intensive energy minimizers and the equilibrium of forces}

Due to the large deviation principles (see section \ref{sec:LDP}),
the stationary distribution of the gas in the thermodynamic limit is obtained through the minimization of the leading term of the energy. In other words, equilibrium distributions are measures that minimize the rate function $I_{\nu}$ and hence enjoy the following property:
\begin{proposition}\label{pro:Equilibrium}
The minimizer $\mu_{\nu}^{\star}$ of the rate function $\I_{\nu}$ is such that each particle is at a classical equilibrium, i.e. the forces acting on each particle cancel out.
\end{proposition}
\begin{proof}
This is a consequence of the fact that the rate function is proportional to the energy. The minimizer of the rate function $\mu_{\nu}^{\star}(q,x)$ satisfies
\[\left . \frac{d\I_{\nu}(\mu)}{d\mu}\right \vert_{\mu=\mu_{\nu}^{\star}} (\varphi)=\int q'\left(g(q')V(x') - 2\int qW(|x'-x|)d\mu_{\nu}^{\star}(q,x) \right)d\varphi(q',x')=0 \]
for any $\varphi$ a signed measure on $Q\times\R^d$ such that $\int\varphi=0$. Therefore, there exists a constant $C$ independent of $x'$ and $q'$ such that $g(q')V(x') - 2\int q W(|x'-x|)d\mu_{\nu}^{\star}(q,x)=C$.
Taking the gradient with respect to $x'$ one obtains
\begin{equation}\label{eq:laplacian}
g(q')\nabla V(x') - 2\int q \nabla_{x'} W(|x'-x|)d\mu_{\nu}^{\star}(q,x)=0\ .
\end{equation}
On the other hand, for an empirical distribution $\mu(q,x) $ the total force acting on particle with charge $q_i$ at position $x_i$ is
\begin{equation}\label{eq:force}
 F_i(\mu(q,x))=-\nabla_{x_i} H(\mu(q,x)) = -N^2\, q_i \left( g(q_i) \,\nabla V(x_i) - 2\int q\nabla_{x_i} W(|x_i-x|)d\mu(q,x)\right)
\end{equation}
and therefore $F_i(\mu_{\nu}^{\star}(q,x))=0$ for all $i$.% \sout{The fact that there is no other distribution for which the forces cancel out is a direct consequence of the uniqueness of the minimizer of the rate function}. 
\end{proof}

This simple and intuitive property provides a convenient method to evaluate the minimizer of the rate function through the analysis of the forces acting on the particles. This equation yields complete characterization of the macroscopic stationary distribution. 

The analysis of sub-leading terms of the energy provides further information on the asymptotic distributions of the heterogeneous gas. Specifically, it allows characterizing the microscopic configurations of the gas at equilibrium. However, this holds in an averaged sense (annealed): the closed-form formulae of the energy hold for the distribution $\rho_q(x)=\int_Q q\mu^{\star}_{\nu}(q,x)$. We discuss here the annealed and quenched microscopic arrangements of the particles: we will show that though the averaged lattice in the situation $g=1$ is a regular Abrikosov lattice, quenched configurations are disordered. In the cases where $g$ is strictly monotonic, quenched and annealed arrangements are identical, and correspond to what will call pseudo-regular progressive lattices. 

We distinguish three cases: (i) constant $g$, (ii) strictly monotonic $g$ for \emph{multi-component} gases (gases whose charges belong to a finite set), and (iii) strictly monotonic $g$ for gases with continuous charge distributions.

\subsection{Constant $g$: Multiple equilibria, disordered gases and disordered lattices}\label{sec:Constantg}
We consider in this section that $g$ is constant. We start by characterizing the macroscopic properties of the gas before characterizing the microscopic arrangements of the charges 
\subsubsection{Equilibria of the rate function} 
In this case, formula~\eqref{eq:force} ensures that a particle with charge $q$ will be at equilibrium anywhere on the support of the measure: the rate function does not constrain the position of the particles depending on their charge. Actually, from~\eqref{eq:laplacian} we have:
\begin{equation}\label{eq:ConditionGConstant}
	\rho_q^{\nu} (x) = \frac 1 {k_d} \Delta V(x),
\end{equation}
which ensures that $\rho_q(x)$ is has uniform density $d/k_d$ on its support. The support of this distribution is hence the ball of radius 
\[R=\left(\frac{k_d \langle q \rangle}{d \vert B_d\vert}\right)^{\frac 1 d}.\]
Moreover, all double layer measures $\mu$ with charge distribution given by~\eqref{eq:ConditionGConstant} have the same value of the intensive energy, since for $g$ constant, this quantity can be written as:
\[h(\mu)=\int_{\R^d} V(x) d\rho_q(x) - \int\int_{\R^d\times\R^d\setminus D} W(\vert x-x'\vert) d\rho_q(x)d\rho_q(x').\]
In this situation the rate function does not have a unique minimum. Depending on the initial distribution of the charges, the system may converge to different equilibria that have the same macroscopic energy. Moreover, they also have the same renormalized energy, since formula~\eqref{eq:EnergyDev} only depends on $\rho_q$.

Among the minimizers, some are perfectly ordered (the values of the charge are monotonic with respect to the radius), some are partially ordered. A unique distribution corresponds to a \emph{disordered} case in which charge and position are independent. This is given by the distribution $\nu(q)\otimes\rho^{\star}(x)$ where $\rho^{\star}$ is the circular law of radius $R$. This distribution is of maximal entropy since the it maximizes the possible microscopic configurations compared to any distribution with a correlation between charge and position. Actually, the number of microscopic configurations in the disordered case is overwhelming compared to any other ordered or partially ordered configuration. 
If the initial configuration is disordered or the temperature is larger than zero, the system will randomly wander on the space of minima (and surroundings) and disordered configurations will be more likely than ordered ones. This is what we observed in the middle plot of Fig.~\ref{fig:transition}.

\subsubsection{The disordered lattice}
Since at equilibrium the charge is not a deterministic function of the position, the renormalized energy does not provide information about the microscopic arrangements for one given realization of a heterogeneous gas~\footnote{We recall that the averaged renormalized energy corresponds exactly to that of an homogeneous gas with charge $q$, which in homogeneous gases allows to characterize microscopic arrangements of particles as done in~\cite{sandier-serfaty:12,rougerie-serfaty:13}. } as discussed in section~\ref{sec:RenormalizedEnergy}. 

%This provides an average result: since heterogeneous Coulomb gases with constant $g$ are distributed as the homogeneous Coulomb gas (with constant charge $\avg{q}$) and  at each location the averaged charge is precisely equal to $\avg{q}$, it is not surprising that the averaged minimal arrangement is similar to that of an homogeneous Coulomb gas. 

Actually, the microscopic arrangement for one given realization (termed \emph{quenched} arrangement) depends on the specific shape of the minimizer. In the, most likely, disordered case, charge and position of the particles are independent and the density of the gas is $\nu(q)\otimes \rho(x)$ where $\rho$ is a circular law. In that case, the quenched microscopic arrangement of the particles will not be regular since the interactions that govern these arrangements are heterogeneous in space due to the charges heterogeneity. The quenched lattice obtained is therefore irregular, and the precise spacing between the particles is correlated to the charges of the particles. %While a regular lattice yields a regular two-points correlation function with a clear spacing appearing, we observe that the correlation is much more disordered and correlations vanish very fast. 
This is visible in a simple 
case when considering the statistics of the minimal distances between charges at zero temperature for gas made of two different charges $q_1<q_2$ (Fig.~\ref{fig:Trimodal}). We clearly observe three typical distances emerging: the typical distance between two particles of charges $q_1$ (small distances), that between particles of charge $q_2$ (large distances), and that between particles of charges $q_1$ and $q_2$ (intermediate distances). The relative representation of these distances vary as a function of the proportion of particles of each charge $q_1$.
\begin{figure}[h]
	\centering
		\subfloat[Scatter plot (90\% of $q_1$)]{\includegraphics[height=5cm]{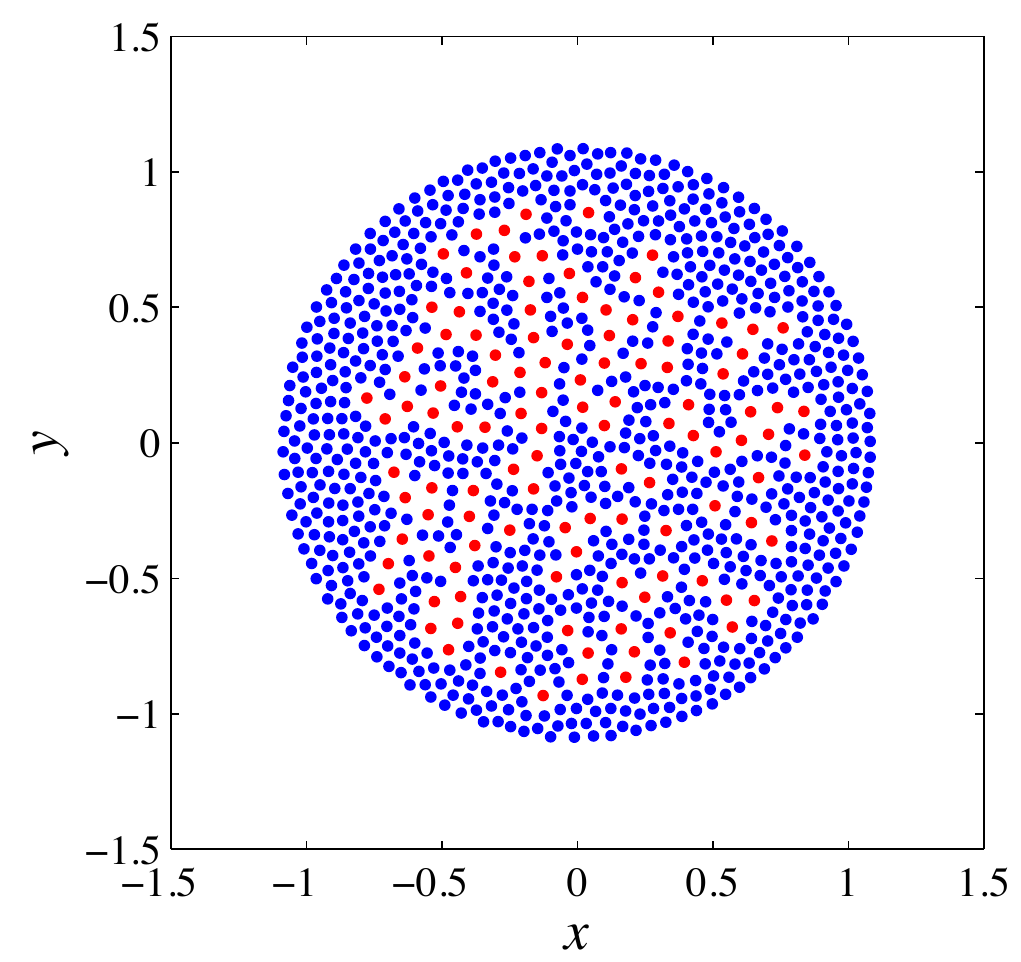}}\qquad \qquad 
		\subfloat[Nearest neighbor statistics]{ \includegraphics[height=5cm]{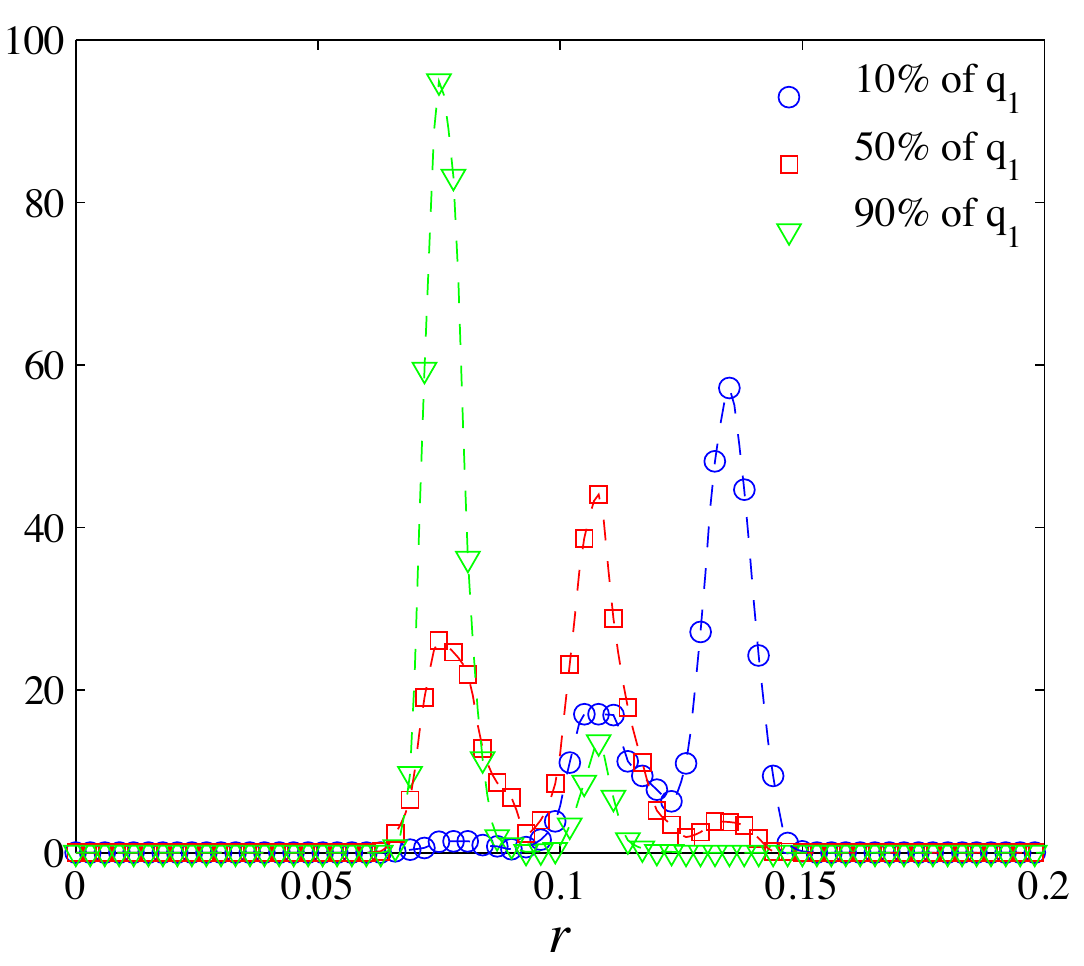}}
	\caption{The disordered lattice: (a) scatter plot of a two-components gas ($90\%$ of particles with charge $q_1=1$, $10\%$ with $q_2=3$) with $g(q)=1$ and $N=1000$. We observe the disorder in the microscopic organization of the charges. (b) Distribution of the distance to the nearest neighbor for the same two-components gas ($N=500$, 8 independent realizations) and different proportions of each type. For a majority of charges $q_1$ (resp. $q_2$) the histogram peaks at a small (resp. large) distance, while for equal proportions the peak arises at the intermediate distance. }
	\label{fig:Trimodal}
\end{figure}
The disorder in the microscopic configurations is also visible from the analysis of the local two-points correlation functions $G(r_0,r)$ of the blow-up configuration (see Fig.~\ref{fig:Correlation}). This quantity provides the probability density of finding a particle at a distance $r$ away from a particle at $x_0$ with $\vert x_0\vert = r_0$.
% :
% \[G(r_0,r)=\frac 1 {N^2 r^{d-1} Z(r_0)} \sum_{i,j=1}^N \delta(\vert x_i\vert-r_0) \delta(\vert x_i-x_j\vert-r)\]
% where $Z$ is a normalization constant equal to $\frac 1 N \sum_{i,j=1}^N \delta(\vert x_i\vert-r_0)$. Here, the $\delta$ function is equal to $\mathbbm{1}_{[-\delta_x,\delta_x]}/\delta_x$ for a small $\delta_x$.
The obtained local correlations functions for $g=1$ show two phenomena: (i) the function does not depend on $r_0$ and (ii) the spatial scale of the damping of the correlations is small (much smaller than that of the homogeneous gas), illustrating the fact that the microscopic arrangements for a given realization are much less correlated, or in other words, much more disordered.  

% Let us eventually note that, if initial conditions are ordered (i.e., monotonically varying with respect to the position), then progressive lattices, in the flavor of what will be shown for monotonic $g$, will be obtained. 
% The precise characterization of the properties disordered quenched lattices is still largely open. Methods developed for homogeneous Coulomb gases are essentially averaged and it is highly non-trivial to extend to the case of heterogeneous gases with constant $g$. 

\subsection{Multicomponent gases and monotonic $g$}\label{sec:MultiComp}
\subsubsection{Macroscopic Equilibrium}
When the map $g$ is strictly monotonic, the sign of the force acting on the particle $i$ depends on $q_i$ and hence so does the equilibrium position. This dependence induces strong correlations between the charge and the position of the particles. %Indeed, if $g$ is strictly monotonic, two particles with distinct charges cannot be at equilibrium on the same sphere. 
Indeed, for $g$ strictly increasing, if a particle of charge $q$ is at equilibrium on the sphere of radius $r$, particles with larger charge will experience stronger confinement and will be pulled towards the origin, while particles with smaller charge will be pushed away from the origin. The converse happens when $g$ is strictly decreasing. As a consequence, at equilibrium, the particles are ordered with respect to their charges, and particles on the surface of a given sphere have all the same charge value.

The precise shape of the particle and charge densities depends strongly on $\nu$. If the charges take discrete values $q_1<\cdots <q_m$ with distinct proportions $(\nu_i, i=1\cdots m)$, i.e. for a charge distribution equal to $\nu(q)=\sum_{i=1}^m \nu_i\delta_{q_i}$, the particles can be classified into a finite number $m$ of populations according to their charge. Because of the spherical symmetry of the problem, the different populations form concentric spherical shells of increasing (decreasing) charge for $g$ strictly decreasing (resp. increasing). Because of the properties of the confinement and interaction potentials in Coulomb gases, the shell associated to population $i$ has a uniform particle density $\frac{d \, g(q_i)}{k_d q_i}$. Indeed, using~\eqref{eq:CoulombCondition} and \eqref{eq:laplacian} one obtains:
%\[\frac{dH(\mu)}{d\mu}=N^2\int q\left(g(q)V(x) - \int q'W(x-y)yd\mu(q',x') \right)d\varphi(q,x)=0\]
%for any $\varphi$ a signed measure on $\R^d$ such that $\int\varphi=0$ so the equilibrium condition for $q=q_i$ is
%\[g(q_i)\,V(x)-\int_Q\int_{\R^d} q' W(x-y) \mu(q',y)dq' dy = C\]
%where $C$ is a constant independent of $x$. Taking the Laplacian of this expression with respect to $x$ and using~\eqref{eq:CoulombCondition} one readily obtains:
\[g(q_i) d = k_d \int_Q q \mu(q,x)dq = k_d\,q_i\,\rho(x).\] 
In order to completely describe the charge distribution, we only need to compute the radii of the shells corresponding to each population. Using the fact that charges are ordered, the inner radius of the shell for population $i$, noted $r^-_i$, can be readily found using equation~\eqref{eq:force}. The outer radius, noted $r^+_i$, is found using the fact that the fraction of particles of charge $q_i$ is equal to $\nu_i$. We obtain the following expressions:
\begin{equation}\label{eq:rminus}
r^-_i= \left(\frac{c_d}{ g(q_i)}\sum_{j\in \mathcal{J}}q_j\nu_j\right)^ {\frac 1 d}\ , \hspace{2cm} r_i^+=\left((r^-_i)^d +\frac{q_i\, \nu_i\, k_d}{d\,g(q_i)\, \vert\mathbb{B}_d\vert} \right)^{\frac 1 d}
\end{equation}
where $\mathcal{J}=\{i+1,\dots,m\}$ for $g$ increasing, and $\mathcal{J}=\{1,\dots,i-1\}$ for $g$ decreasing.
Interestingly, the charges are strictly segregated by empty shells. Indeed, for $j=i+1$ ($j=i-1$) for $g$ strictly decreasing (resp. increasing):
\[(r_i^+)^d - (r_{j}^-)^d = (r_i^-)^d \left(1-\frac{g(q_i)}{g(q_{j})}\right) + \frac{\nu_i\,q_i}{g(q_{j})}\left( \frac{k_d}{d\,\vert \mathbb{B}_d\vert}-c_d\frac{g(q_i)}{g(q_j)} \right) <0\]
because (i) $\frac{g(q_i)}{g(q_{j})}>1$ hence the first term is strictly negative and (ii) the geometric constants are such that 
$\frac{k_d}{d\,\vert \mathbb{B}_d\vert}=c_d \frac{\vert \S_{d-1}\vert}{d\vert \mathbb{B}_d\vert} = c_d$.

Figure~\ref{fig:atom} shows that numerical simulations of the stationary distribution of a finite-sized multicomponent gas clearly display the concentric separated shell structure predicted theoretically with borders that are in good quantitative agreement with the analytic expressions.

\begin{figure}[h]
\centering
\includegraphics[width=0.7\textwidth]{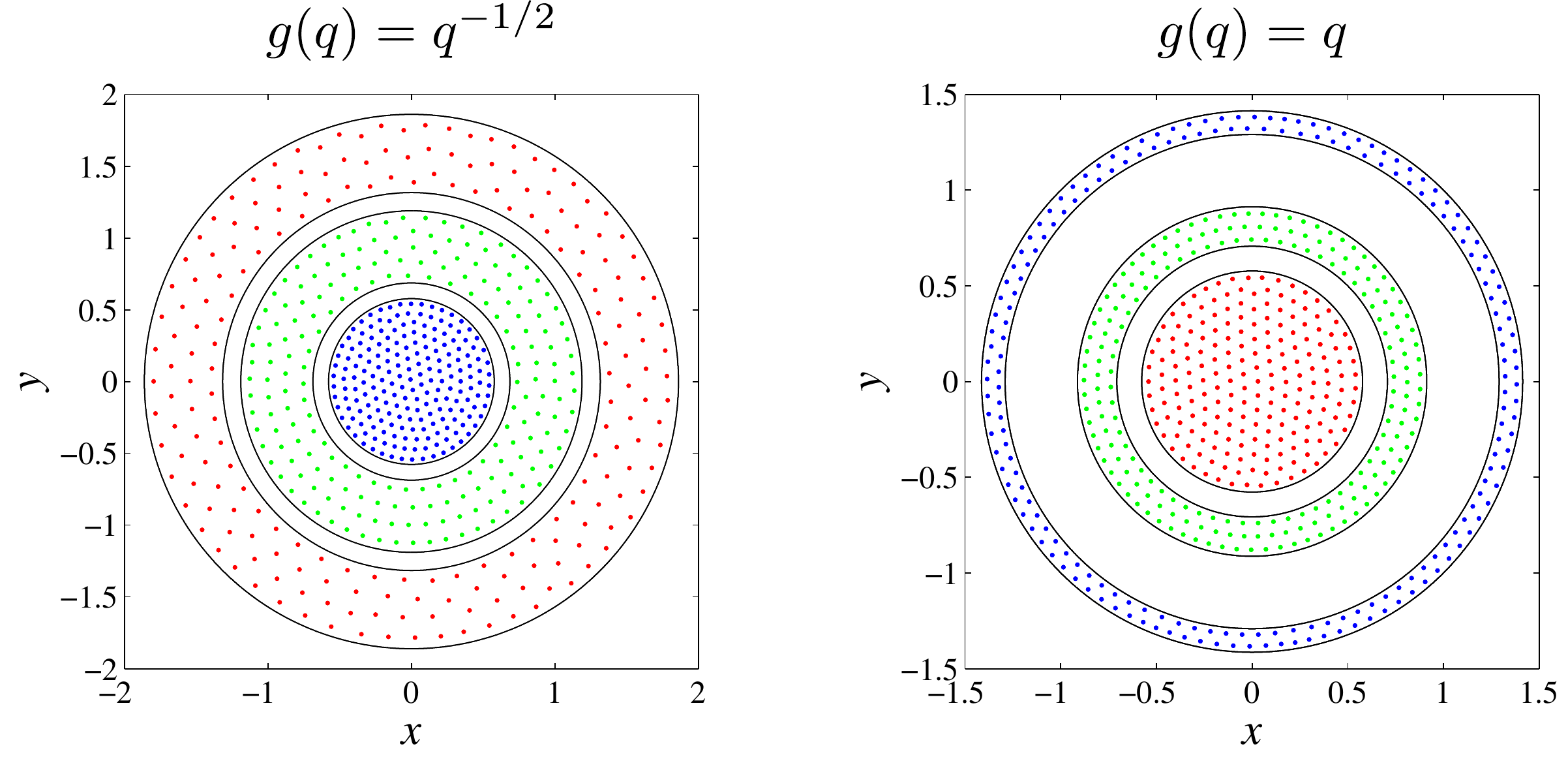}
 \caption{Scatter plot of the stationary distribution obtained from the simulation of a two dimensional Coulomb gas with $N=600$ particles, where $1/3$ of the particles have unit charge (blue), $1/3$ have charge two (green) and $1/3$ have charge three (red). The particles separate into three disjoint shells according to their charge. The black lines are the theoretical values of the borders  $r^-$ and $r^+$ for each shell for an infinite gas (see text for the expressions).}\label{fig:atom}
\end{figure}

\subsubsection{Microscopic arrangements}
Since particles arrange into disjoint shells on which the distribution is uniform and charge is constant, the minimized energy in each shell is analogous to that of homogeneous Coulomb gases with charge given by the charge of the shell.
%, they shall form Abrikosov lattices with distinct spacings depending on the charge in each shell, increasing with the radius both when $g$ is decreasing or increasing. 
In detail, for quantified charges $(q_j, j\in \mathcal{J})$, the distribution being equal to:
\[\mu_{\nu}^{\star}(q,x)=\sum_{j\in\mathcal{J}} \delta_{q_j} \nu_{j} \mathbbm{1}_{r_{j}^- \leq \vert x \vert \leq r_j^+ },\]
the order $N$ term of the energy (expressed through the renormalized energy functional) is the sum over all shells of the renormalized energy in each shell, as visible by injecting the above formula into~\eqref{eq:ClosedForm}. Therefore, the microscopic configuration corresponds to arrangements that minimize the renormalized energy in each shell, conjectured to be the Abrikosov triangular lattice in the two-dimensional Coulomb gas. The spacing between particles however depends on the charge of the shell considered, and progressively changes with the radius. Multi-component two-dimensional Coulomb gases therefore produce disjoint mixtures of Abrikosov lattices. In higher dimensions, the arrangement is also identical to that of an homogeneous gas in each shell, conjectured to be a regular lattice. 

\subsection{Continuous charge distributions}
\subsubsection{Macroscopic distributions}
We now consider a gas with a charge distribution $\nu$ absolutely continuous with respect to Lebesgue's measure (with no atoms), with a continuous support
\footnote{For a measure with a non-connected support, a combination of the arguments of multicomponent gases with the continuous support case can be applied.}. 
This case can be heuristically seen as a continuous limit of the multicomponent gas. In that limit, the shells become increasingly fine and close to each other, yielding a continuous particle and charge distribution, and a smooth dependence of the charge with respect to the radius. In particular, spherical symmetry and the properties of Coulomb interaction allow to compute the force~\eqref{eq:force} and obtain the following implicit equation describing the equilibrium of a particle with charge $q_i$ at position $r_i$:
\begin{equation}\label{eq:eqcondition}
 r_i=\left(\frac{c_d}{ g(q_i)}\int_{0}^{r_i}\rho_q^{\nu}(r)r^{d-1}dr \right)^ {1/d}.
\end{equation}
Moreover, the charge of the particles at location $r$, $q(r)$, is a continuous and strictly monotonic map, therefore invertible. We denote by $r(q)$ its inverse. This allows to solve analytically the implicit equation~\eqref{eq:eqcondition}. First of all, using the change of variables formula, the particle density can be expressed in terms of the charge distribution $\nu(q)$ and $q(r)$ as 
\begin{equation}\label{eq:changeofvariables}
 \rho(r)r^{d-1}dr=\nu(q(r))|q'(r)|dr\ .
\end{equation}
The radial charge density expresses simply as:
\[\rho_q(r)=q(r)\rho(r)\ .\] 
Using equation~\eqref{eq:eqcondition} together with these two relationships, we obtain the explicit expression for $r(q)$ as a function of the parameters of the model:
\begin{equation}\label{eq:rofq}
 r(q)=\left(\frac{c_d}{ g(q)}\int_{q_-}^{q_+} u\nu(u)du\right)^ {1/d}
\end{equation}
where $[q_-,q_+]$ is the interval $[q_{\min},q]$ for $g$ strictly decreasing, and $[q,q_{\max}]$ for $g$ strictly increasing. %\footnote{This expression is exactly the continuous limit of the multicomponent gas~\eqref{eq:rminus}.}. 
This map is therefore strictly monotonic in both cases, and one can recover $q(r)$, therefore compute $\rho(r)$ and $\rho_q(r)$. These expressions of the density are valid only in a ball of finite radius $R$ which is given by
\[ R=\left(\frac{c_d}{ g(q^{\star})}\avg{q}\right)^ {1/d}\]
where $g(q^{\star})=\min(g(q_{\min}),g(q_{\max}))$. Outside of the ball thei are both equal to $0$.

This analytical result is illustrated by numerical simulations (Fig.~\ref{fig:distributions}), performed with a uniform charge distribution $\nu(q)$ over the interval $[q_{\min},q_{\max}]$, both in the case of increasing functions ($g(q)=q$) and a decreasing function ($g(q)=1/\sqrt{q}$). In the case of $g(q)=q$, these formulae greatly simply, and denoting $\delta q = q_{\max}-q_{\min}$ one obtains:
\begin{equation*}
 \begin{cases}
 	r(q) &= \displaystyle{\left(\frac{c_d(q_{\max}^2-q^2)}{2\,q\delta q}\right)^ {1/d}}\\
	R    &= \displaystyle{\left(\frac{c_d(q_{\max}+q_{\min})}{2\,q_{\min}}\right)^ {1/d}}
 \end{cases}
\end{equation*}
and the expressions of the charge and particle densities, as well as the map $q(r)$, are given by:
\begin{equation*}
 \begin{cases}
 	q(r)    &= \displaystyle{-\frac{r^d\,\delta q}{c_d} + \sqrt{\left(\frac{r^d\, \delta q}{c_d}\right)^2+q_{\max}^2}}\\
\\
	\rho(r)& = \displaystyle{\frac{-d}{k_d}+\frac{d\,\delta q \, r^d}{k_d \, c_d \sqrt{\left(\frac{r^d\delta q}{c_d}\right)^2+q_{\max}^2}}}
\end{cases}
\end{equation*}

\begin{figure}[h]
\centering

\includegraphics[width=0.7\textwidth]{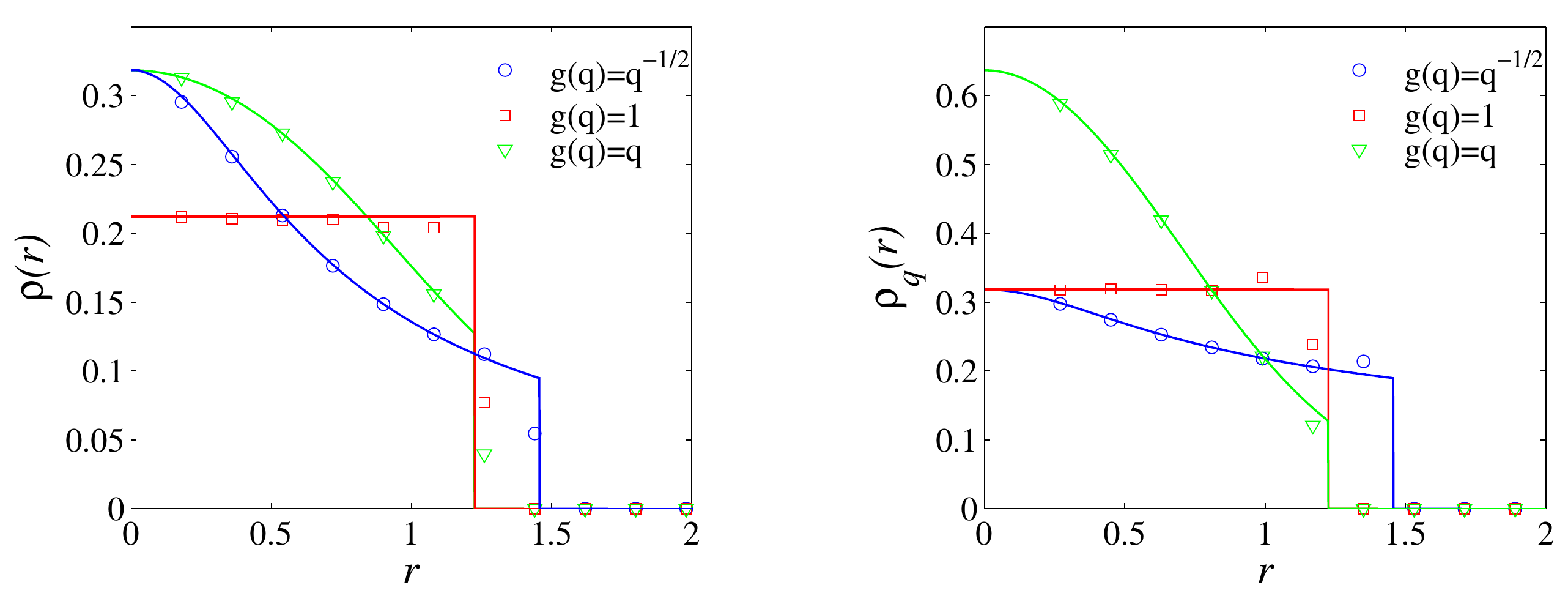}

\caption{Particle and charge distributions for a 2 dimensional Coulomb gas for several choices of $g(q)$. The solid lines are the theoretical predictions, the points correspond to the average of $100$ simulations of the gas. Standard errors are smaller than the points. $N=1000$ and $\nu(q)$ uniform in $Q=[1,2]$.}
\label{fig:distributions}
\end{figure}
\subsubsection{Universality of Heterogeneous Coulomb gases}\label{append:Univ}
We now show that heterogeneous Coulomb gases are universal, in the sense that its equilibrium distribution span over a wide range of radially symmetric distributions:
\begin{thm}\label{thm:universality}
	Let $g:\R_+\mapsto \R_+$ be a monotonic map. 
	\begin{itemize}
		\item For any radially symmetric probability distribution on $\R^d$ with bounded support and decreasing density $f$ along the radial axis, there exists a distribution of charge $\nu$ absolutely continuous with respect to Lebesgue's measure such that $f$ is the equilibrium distribution of the heterogeneous gas with weight function $g$ and charge distribution $\nu$.
		\item the density $\nu$ can be constructed from the analysis of a planar dynamical system. 
	\end{itemize} 
\end{thm}

\begin{proof}
	Let $g$ a fixed monotonic function (to fix ideas, we will assume $g$ strictly decreasing, the same proof applies for $g$ increasing) and consider a radially symmetric measure with radial density $f(r)$. We construct a charge density $\nu(q)$ such that the empirical measure of the heterogeneous gas with charges distributed as $\nu$ converges to $f(r)$.
	
	Equation~\eqref{eq:rofq} can be rewritten in the form:
	\begin{equation}\label{eq:xiofq}
	\xi(q):=\Big(r(q)\Big)^d=\frac{c_d}{g(q)}\int_{q_{min}}^q u\nu(u)du 
	\end{equation}
	and we have shown that this relation is invertible. We denote with a slight abuse of notations $q(\xi)$ its inverse and $f(\xi)$ the composed function $f(\xi(r))$. 
	
	Possible charge densities $\nu$ yielding a particle distribution $f$ satisfy the relationship:
	\begin{equation}\label{eq:normalization}
		\nu(q(\xi))=\vert \S_{d-1}\vert f(\xi) r^{d-1} \left\vert\frac{dr}{dq}\right\vert = \frac{\vert \S_{d-1}\vert}{d} f(\xi) \left|\frac{d\xi}{dq}\right|
	\end{equation}
	and using equation~\eqref{eq:xiofq}, we obtain the necessary condition
	\[dq=\frac{1-a(q)f(\xi)}{-b(q)\xi}d\xi\]
	where $a(q)=\frac{\vert \S_{d-1}\vert}{d} \frac q{g(q)}$ and $b(q)=\frac{g'(q)}{g(q)}$. This provides an ordinary differential equation on $\xi(q)$, with complex dynamics. However, trajectories can be found as the solutions in the phase plane of the two-dimensional dynamical system:
	\begin{equation}\label{eq:DynSyst}
	 \begin{cases}
	  \frac{dq}{dt}=F(q,\xi)=1-a(q)f(\xi)\\
	  \frac{d\xi}{dt}=G(q,\xi)=-b(q)\xi\ .
	 \end{cases}
	\end{equation}
	Solutions to this equation provide a set of maps $\xi(q)$. Since we assumed $g$ decreasing, we need to have $\xi(q)$ increasing, which is possible only when $f$ is decreasing. The thus defined dynamical system has a unique fixed point at $\xi=0$ and $q=q_{\min}$ defined by $f(0) a(q_{\min})=1$ (uniqueness comes from the fact that fixed points in $q$ are defined by the intersection of the identify and the map $g(q)/f(0)$ which is decreasing). Left panel of Figure~\ref{fig:DynSyst} shows a representation of the phase plane of the system.
	\begin{figure}[h]
	    \centering
	    \includegraphics[width=0.7\textwidth]{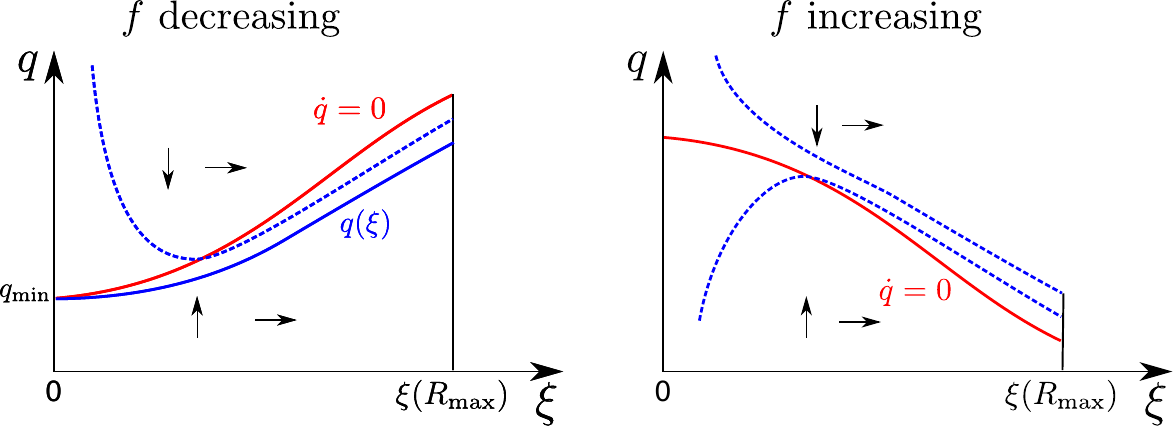}
	    \caption{Schematic represention of the phase plane of~\eqref{eq:DynSyst}. Arrows indicate the direction of the vector field, and red curve is the $q$-nullcline $f(\xi) a(q)=1$. (Left) $f$ decreasing. Dashed blue: non-increasing solution. Solid blue: unstable manifold of the saddle, providing the acceptable solution $\xi(q)$ (see proof of theorem~\ref{thm:universality}). (Right) $f$ increasing. All solutions are non-increasing (see corollary~\ref{cor:Decrease}).}
\label{fig:DynSyst}
\end{figure}
	
	The second point is shown as follows. The Jacobian matrix at the unique fixed point  $(q_{\min},0)$  reads:
	\[\Matrix{-b(q_{\min})}{0}{f'(0)/f(0)}{-a'(q_{\min})f(0)}.\]
	Since we have $b(q)<0$ and $a'(q)>0$, the fixed point is a saddle. Its stable manifold is the line $\xi=0$ (hence any initial condition with $\xi=0$ yield trajectories such that $\xi(t)=0$ for all times and $q$ converges to $q_{\min}$). The unstable manifold of the fixed point is a solution of the problem. It is the only solution, since any solution with initial condition $\xi>0$ not on the unstable manifold either diverge or become negative for the backward integration, before reaching $\xi=0$. This ends the proof of the second point.  
	
	We eventally notice that $q_{\max}$ is determined by the value of $q$ on the boundary of the support of $f$, and that normalization of $\nu$ is naturally satisfied thanks to relationship~\eqref{eq:normalization}. 
\end{proof} 

The proof of the theorem is constructive. One can implement numerically the construction of the unstable manifold of the fixed point from which closed form formulae yield the charge distribution (see Fig.~\ref{fig:Reconstruction}).
\begin{figure}[h]
  \centering
  \includegraphics[width=\textwidth]{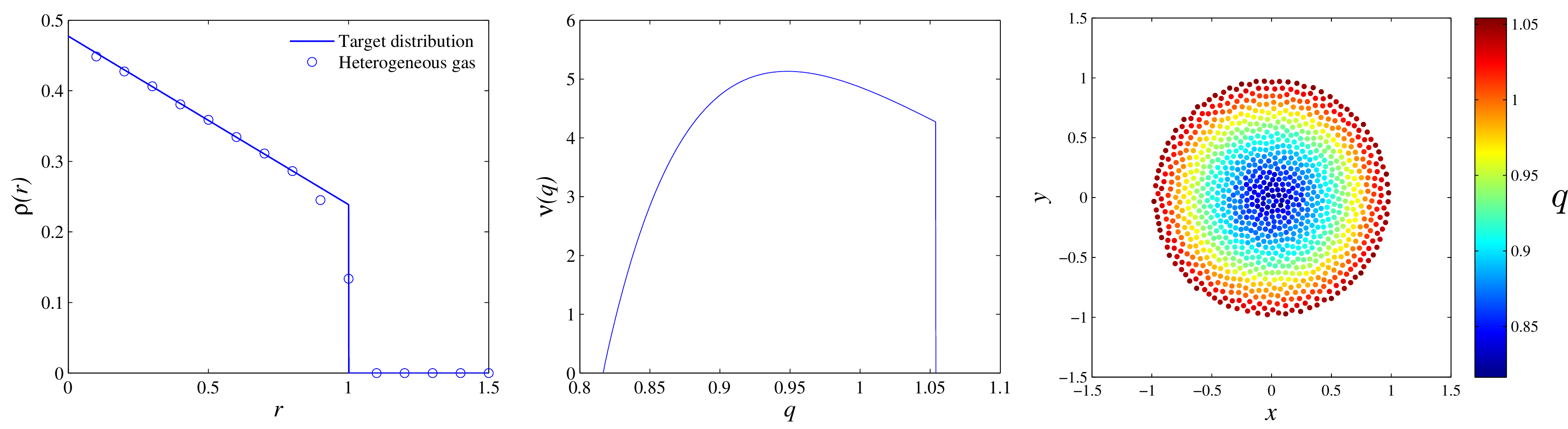}
  \caption{$N=1000$, $g(q)=1/q$. (Left) Solid line: target radial distribution $\rho(r)=\frac{3}{4\pi}(2-r)\,\1_{r<1}$, points: particle distribution of the generated heterogeneous gas. Each point is an average over $100$ realizations. (Center) Charge density $\nu(q)$. (Right) Scatter plot of one realization of the gas.}
\label{fig:Reconstruction}
\end{figure}

As a side result of the demonstration, we can show the following general result on the distribution of particles in heterogeneous gases:
\begin{corollary}\label{cor:Decrease}
 The particle distribution of heterogeneous gases is a decreasing function.
\end{corollary}
 
 \begin{proof}
 The proof of theorem~\ref{thm:universality} provides a characterization of the map $q(r)$ as a function of the particle distribution $f(r)$. When $f$ is strictly increasing and $g$ is decreasing (resp. increasing), the map $q(r)$ is not increasing (resp. decreasing) violating the charge ordering of heterogeneous gases (see right panel of Fig.~\ref{fig:DynSyst}).
\end{proof}

\subsubsection{Microscopic configurations and the progressive lattice}
For continuous charge distributions and $g$ monotonic, heterogeneous gases have a unique equilibrium in which the particles are spatially arranged according to their charge. The equilibrium distribution splits as $\mu_{\nu}^{\star}(q,x) = \delta_{q(x)}\rho(x)$, where $q(x)$ is a monotonic function of $\vert x \vert$. Therefore, as noted in section~\ref{sec:RenormConstantg}, the charge at equilibrium is a deterministic function of the position, and microscopic arrangements minimize our heterogeneous renormalized energy. The monotonicity of the map $q(x)$ will generally yield non-regular lattices, since, heuristically, one necessary condition in order to obtain a regular lattice is the interchangeability between particles, and in particular the local homogeneity of the charges. Indeed, in order for the particles to organize in a perfect equilateral triangular lattice (or perfectly regular lattice in higher dimensions), the pairwise repulsion between charges has to be identical. 

Here, the charge of the particles progressively changes with the radius on $\R^d$ and therefore both confinement and repulsion continuously vary in space. This will produce arrangements that progressively vary as a function of space. This can be also seen from the fact that gases with continuous charge distribution are heuristically the continuum limit of the multi-component gas. The microscopic arrangement produces what we call \emph{progressive lattices}, in the sense that charges keep forming a triangular lattice, but these are no more equilateral: the edge of the triangle varies as a function of the radius of the position of the particle, it increases with the radius both in the case where $g$ increases or decreases, as visible in the scatter plots of Figure~\ref{fig:transition} and in the decreasing shape of the density. This property is further illustrated by the computation of the local two-points correlation function $G(r_0,r)$, presenting a peak at a position $r$ that increases continuously as a 
function of the radius $r_0$ of the shell 
around which the statistics are computed (Fig.~\ref{fig:Correlation}). This indicates that the typical distance between particles increases as their location becomes increasingly remote from the origin. 

\begin{figure}[h]
	\centering
		\includegraphics[width=\textwidth]{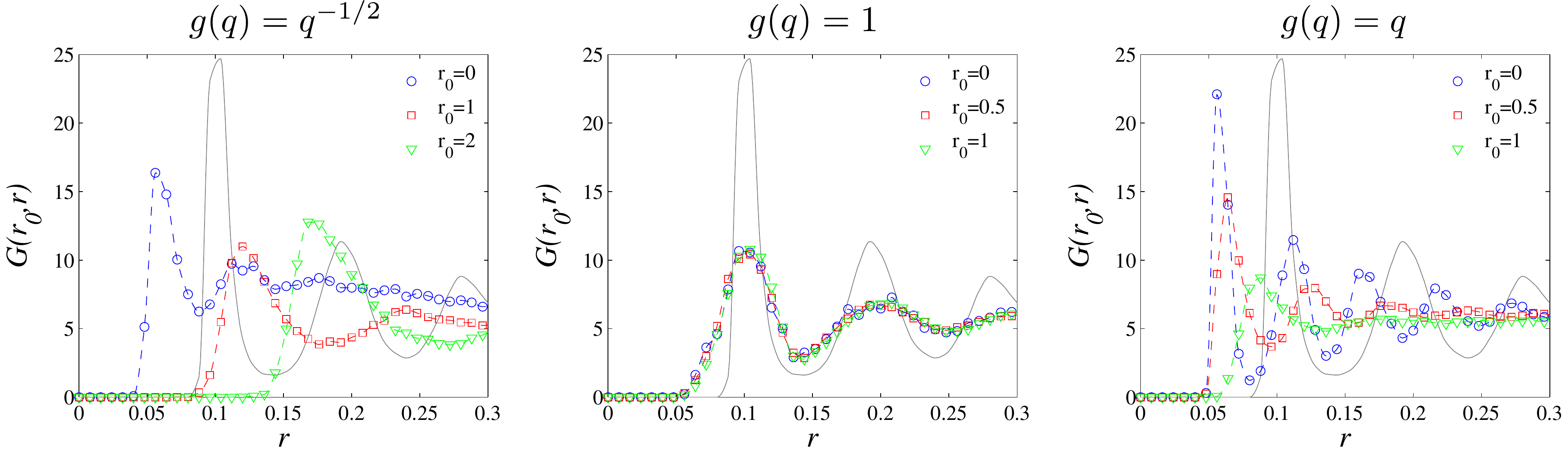}
	\caption{Local two-points correlation function $G(r_0,r)$ for heterogeneous two-dimensional Coulomb gases for different choices of $g$ at different positions $r_0$. Statistics computed over $100$ realization of a $N=1000$ Coulomb gas with $\nu(q)$ uniform in $Q=[1,5]$. Gray solid line corresponds to the correlation function of an homogeneous Coulomb gas (with charge $\langle q\rangle$) for reference.}
	\label{fig:Correlation}
\end{figure}

\begin{remark}
	We fully investigated here the cases of constant and strictly monotonic $g$. These cases have the interest to unfold the degenerate state corresponding to constant $g$ into a unique equilibrium, which in addition shows a regular ordering in space of the particles with respect to their charges and relatively regular lattices. The case of a general map $g$ can be treated along the same lines. The analysis of the forces~\eqref{eq:force} provides, for a given position $x$ in space, a unique value of $g(q)$ at equilibrium. When $g$ is not invertible, this yields a set of values of the charge for which a particle will be at equilibrium at location $x$ (for $g$ monotonic, a unique value, and for $g$ constant, any value). Therefore the equilibrium distribution will not be unique anymore and the maximal entropy solution will display a certain level of disorder. Microscopic distributions are also more complex: similarly to the constant $g$ case, they display disorder (since particle with distinct charges co-exist 
locally), but similarly to strictly monotonic $g$ cases, present a progressive evolution of the mean lattice spacing (the averaged lattice is not regular anymore). An example of such gas is displayed in Fig.~\ref{fig:SineGas}.
\end{remark}
\begin{figure}[h]
	\centering
		\includegraphics[width=\textwidth]{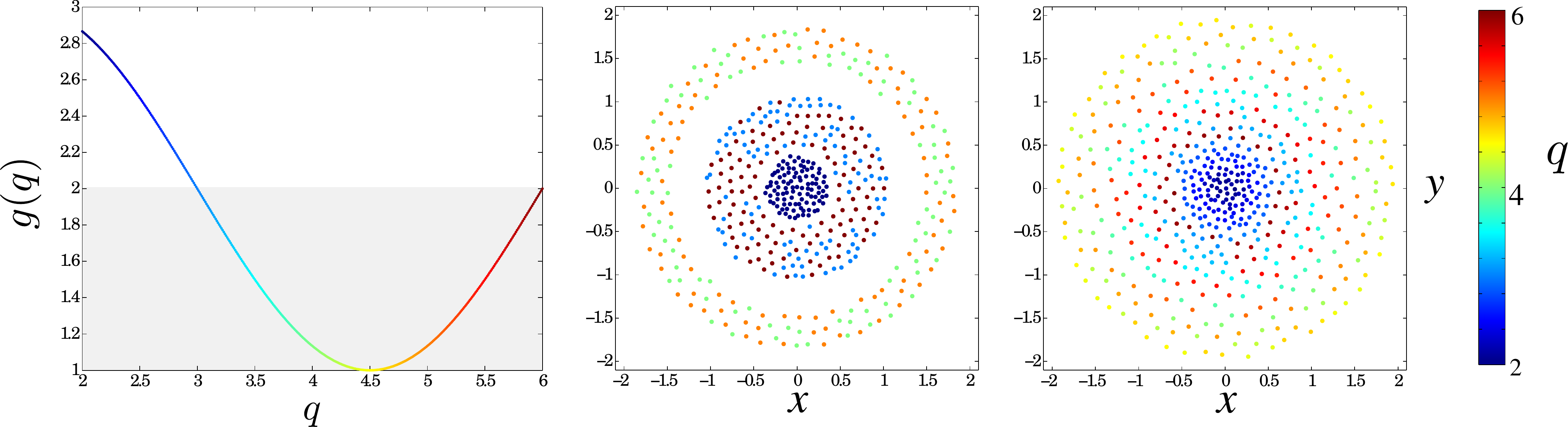}
	\caption{Gas with non-monotonic $g$. (Left) the map $g$ chosen, $g(q)=2+\sin(\pi q/3)$: in the gray box, two values of $q$ correspond to the same $g(q)$. The function $g$ is not one on one.  Well selected values $q$ and $q'$ on either side of its mean satisfy $g(q)=g(q')$. (Center) Multi-Component gas with $N=500$ and uniform charge in $\{2,\cdots,6\}$, (Right) Continuous-charge distribution with $N=500$ and $\nu(q)$ uniform in $Q=[2,6]$. Charges corresponding to $q<3$ are only found at one specific location (in the center). In other shells (radii) there is a mixture of two distinct charges (corresponding to the same value of $g(q)$). Note that the lattices are irregular in these shells, contrasting with the monotonic $g$ case, and related to the mixture between particles with clearly distinct charges. }
	\label{fig:SineGas}
\end{figure}

\appendix
\bigskip\bigskip
\noindent {\Large \bf Appendix}

\section{General heterogeneous gases in $\R^d$}\label{append:GeneralGases}
In this appendix we discuss the generalization of the results demonstrated in this paper to gases with non-Coulomb interaction in $\R^d$. Section~\ref{append:manifolds} explores the case of gases on manifolds.  In the main text, we showed that heterogeneous Coulomb gases display very different behaviors depending on the properties of the map $g(q)$: if the map is constant, the particles are distributed according to the circular law and are disordered with respect to their charges, but as soon as $g(q)$ is monotonic the charge and position of the particles is strongly correlated and the particles distribution is no more uniform. This phenomenon, demonstrated analytically for Coulomb gases, is much more general. In fact, as we show here, it is also valid in arbitrary dimension, and for more general interaction potentials. In order to analyze general interactions we consider potentials defined by $\nabla W(r) \propto r^{-d+1-\eta}$ for $\eta$ a real parameter. These correspond either to stronger (super-Coulomb, 
$\eta>0$) or weaker (sub-Coulomb $\eta<0$) repulsion at small scales. 
Most of the mathematical results proved in the frame of Coulomb gases are valid in these generalized contexts. We start by discussing Large Deviations Principles for such heterogeneous gases and next-to-leading order correction terms of the energy, before discussing their asymptotic behaviors. 

\subsection{Large-deviations principles for general heterogeneous gases}\label{sec:LDPGeneralGases}
We consider gases in $\R^d$ with general confining and interaction potentials satisfying the following assumptions (identical to~\cite{chafai2013first}):
\renewcommand{\theenumi}{(H\arabic{enumi})}
\begin{enumerate}
	\item\label{assumption:H1} The map $W:\R^d\times\R^d \mapsto (-\infty,+\infty]$ is continuous in $\R^d\times\R^d$, symmetric, takes finite values on $\R^ d\times\R^d\setminus \{(x,x);x\in\R^d\}$ and satisfies the following integrability condition:  for all compact subset $K\subset \R^d $, the function
	\[z\in\R^d \mapsto \sup\{W(x,y);|x-y|\geq |z|,x,y\in K\}  \]
	is locally Lebesgue-integrable on $\R^d$.
	\item\label{assumption:H2} The function $V:\R^d\mapsto\R$ is continuous and such that $lim_{|x|\to\infty}V(x)=+\infty$ and
	\[ \int_{\R^d}\exp(-V(x))dx<\infty\ . \]
	\item\label{assumption:H3} There exist constants $c\in\R$ and $\varepsilon_0\in(0,1)$ such that for every $x,y\in\R^d$,
	\[ W(x,y)\geq c-\varepsilon_0(V(x)+V(y)). \] 
\end{enumerate}
\begin{remark}
	These properties are clearly satisfied for sub-Coulomb and Coulomb interactions. They are also valid for super-Coulomb interactions provided that $\eta>-(d+2)$. 
\end{remark}

An additional property is necessary to demonstrate the large-deviation principle for general gases:
\begin{enumerate}
	\setcounter{enumi}{3}
	\item\label{assumption:H4} For all $\mu \in \M_1(Q\times \R^d)$ such that $h(\mu)<\infty$, there exists a sequence of probability measures $\mu_n$, with marginal particle distribution absolutely continuous with respect to Lebesgue's measure on $\R^d$, such that $\mu_n \rightharpoonup \mu$ and $h(\mu_n)\to h(\mu)$. 
\end{enumerate}
This property is valid for sub-Coulomb interactions, as shown in~\cite{chafai2013first}, as these fall in the class of Riesz gases. 

Under these assumptions, large-deviations principle can be proved in the case of homogeneous gases, and can be extended to heterogeneous gases conditioned on the marginal charge distribution:
\renewcommand{\theenumi}{(\roman{enumi})}
\begin{corollary}[of Theorem~\ref{thm:LDP_Coulomb}]\label{cor:GeneralGases}
Under assumptions~\ref{assumption:H1}-~\ref{assumption:H3}, we have:
\begin{enumerate}
	\item The function $h$ is lower semicontinuous, has compact level sets and $\inf_{\M_1(Q\times \R^d)}h >-\infty$. 
	\item Moreover, under assumption~\ref{assumption:H4}, we have:
	\begin{equation*}
		\begin{cases}
			\limsup_{N\to\infty}\;\frac 1 {\beta N^2}\log\P[\hat{\mu}_N \in \M_{\hat\nu_N} \cap B_{FM}(\mu,\delta)] \leq -\mathcal{I}_{\nu}(\mu)\\
			\ \\
			\displaystyle{\lim_{\delta \searrow 0} \; \liminf_{N\to\infty} \; \frac 1 {\beta N^2}\log\P[\hat{\mu}_N \in \M_{\hat\nu_N} \cap B_{FM}(\mu,\delta)] \geq -\mathcal{I}_{\nu}(\mu)}\ .
		\end{cases}
	\end{equation*}
	where $B_{FM}(\mu,\delta)$ is the Fortet-Mourier ball of radius $\delta$ centered at $\mu$ and $\M_{\hat\nu_N}$. In particular, denoting $\I_{\nu}^{\min} = \{\mu\in \M_{\nu}, \I_{\nu}(\mu)=\min_{\M_{\nu}} I_{\nu}\}$, we have:
	\[\lim_{N\to\infty} d_{FM}(\hat{\mu}^N, \I_{\nu}^{\min})=0. \]
	\item For sub-Coulomb gases, assumption~\ref{assumption:H4} is naturally satisfied, and moreover the map $h$ is strictly convex. There exists therefore a unique minimizer of $\I_{\nu}$, denoted $\mu^{*}_{\nu}$, and we have the almost sure convergence:
	\[\lim_{N\to\infty} d_{FM}(\hat{\mu}^N, \mu^{*}_{\nu})=0.\]
\end{enumerate}
\end{corollary}

In contrast with the 2 dimensional Coulomb gas (and with the usual Sanov theorem), this result is proved when the space of probability measures is metrized by the Fortet-Mourier distance:
	\[d_{FM}(\mu,\tilde{\mu}) = \sup\limits_{\vert f\vert_{\infty}, \vert f\vert_{\text{Lip}}<1}\left\{\int f d\mu - \int f d\tilde{\mu}\right\},\] 
	which is compatible with the weak topology. It remains an open problem to show that the large-deviations principles for general gases is valid in the stronger Wasserstein topology. 
 
Whether assumption~\ref{assumption:H4} holds or if the rate function is convex for a given interaction potential remains an open problem. Tools from potential theory~\cite{chafai2013first,landkoffoundations} are useful in order to show that these properties hold for certain interaction kernels. Even if the rate function is strictly convex in the homogeneous gas case, this is not necessarily the case of heterogeneous gases. In particular, it is clear that here again, the rate function is constant on the space of probability measures:
\[\mathcal{M}_{\rho_1,\rho_2}=\Big\{\mu \in \M_1(Q\times \R^d)\; ;\; \int_Q q g(q)d\mu(q,x) =\rho_1(x),\; \int_Q q d\mu(q,x) =\rho_2(x)\Big\}.\]

\subsection{The splitting formula for general gases}
The methods developed in order to find next-to-leading order terms of the $N$-particles energy and their relationship with the heterogeneous renormalized energy function strongly rely on the fact that the interaction kernel, in Coulomb gases, is the Green function of the Laplace operator. While this is valid in any dimension, such simplification will not occur in the case of sub- or super-Coulomb interactions. The characterization of the next-to-leading order terms can nevertheless be performed through a generalized splitting formula, demonstrated for Coulomb gases in theorem~\ref{thm:splitting}:
\begin{corollary}[Sub-leading terms for generalized gases]\label{thm:SplittingGeneral}
	For general gases, sub-leading corrections of the $N$ particles energy can be expressed as a function of the empirical measure, and one obtains: 
\[H_N=N^2 \I(\mu_{\nu}^{\star}) + 2N\int_{Q\times\R^d} \zeta(q,x) d\delta_N(q,x)-\int_{Q\times Q \times D^c} qq' W(\vert x-x'\vert) d\delta_N(q,x) d\delta_N(q',x')\]
	with 
	\[\zeta(q,x) = \frac{q g(q)}{2} V(x)- \int_{Q\times\R^d} q q' W(\vert x-x'\vert) d\mu_{\nu}^{\star}(q,x)\]
	which vanishes on the support of $\mu_{\nu}^{\star}$.
\end{corollary}
\begin{proof}
	This is nothing else than the formula~\eqref{eq:EnergyDev} obtained as a side result in the proof of theorem~\ref{thm:splitting}. There was no use of the specific properties of Coulomb gases up to this point. 
\end{proof}
Note that for specific interaction kernels, for instance when $W$ is the Green function of an operator (e.g., Riesz gases for which the interaction kernel is the Green function of the fractional Laplace operator), further simplifications may lead to the introduction of a specific generalized functional, in the form of the renormalized energy.

\subsection{Macroscopic and microscopic distributions of charges}
The characterization of the minima of the rate function for Coulomb gases was performed using a classical mechanics argument based on the equilibrium of forces. Coulomb interactions allow great simplification of the total force resulting from the repulsion of all charges on a given particle which allowed to uncover closed-form expressions for the equilibrium distributions. No similar simplification arises for general gases. However, we can demonstrate the same ordering transition for general gases (see Fig.~\ref{fig:generalgases}).

\begin{proposition}\label{pro:GeneralGasesMacro}
	General gases in $\R^d$ undergo the following ordering transition:
	\begin{enumerate}
		\item for $g$ strictly monotonic, particles are ordered with respect to their charge, and the ordering depends on the sense of variation of $g$
		\item for constant $g$, infinitely many distributions minimize the rate function, among which the disordered distribution $\nu(q)\otimes\rho(x)$ where $\rho$ is the equilibrium distribution of the homogeneous gas with charge $\langle q\rangle$.
	\end{enumerate}
\end{proposition} 

\begin{proof}
	Proposition~\ref{pro:Equilibrium} is valid for general gases, and characterize the equilibrium distributions. In particular, formula~\eqref{eq:force} ensures that if a particle with charge $q$ is at equilibrium at location $x$, every particle with a distinct charge can not be at equilibrium at $x$. If $g$ is increasing (resp. decreasing), larger charges will be pushed towards (resp. pulled away from) the origin, proving point (i). 
	
	In the case where $g$ is constant, we observe that the force experienced by a particle at location $x$ is independent of the charge. Actually, for constant $g$, the rate function only constrains the value of $\rho_q$. This charge distribution will belong to the minima of the rate function of a virtual homogeneous gas with charge $\langle q \rangle$, and any double-layer distribution $\mu$ with charge distribution $\rho_q$ minimizes the rate function. In particular, the disordered distribution does minimize the rate function. 
\end{proof}

\begin{figure}
	\centering
		\includegraphics[width=0.85\textwidth]{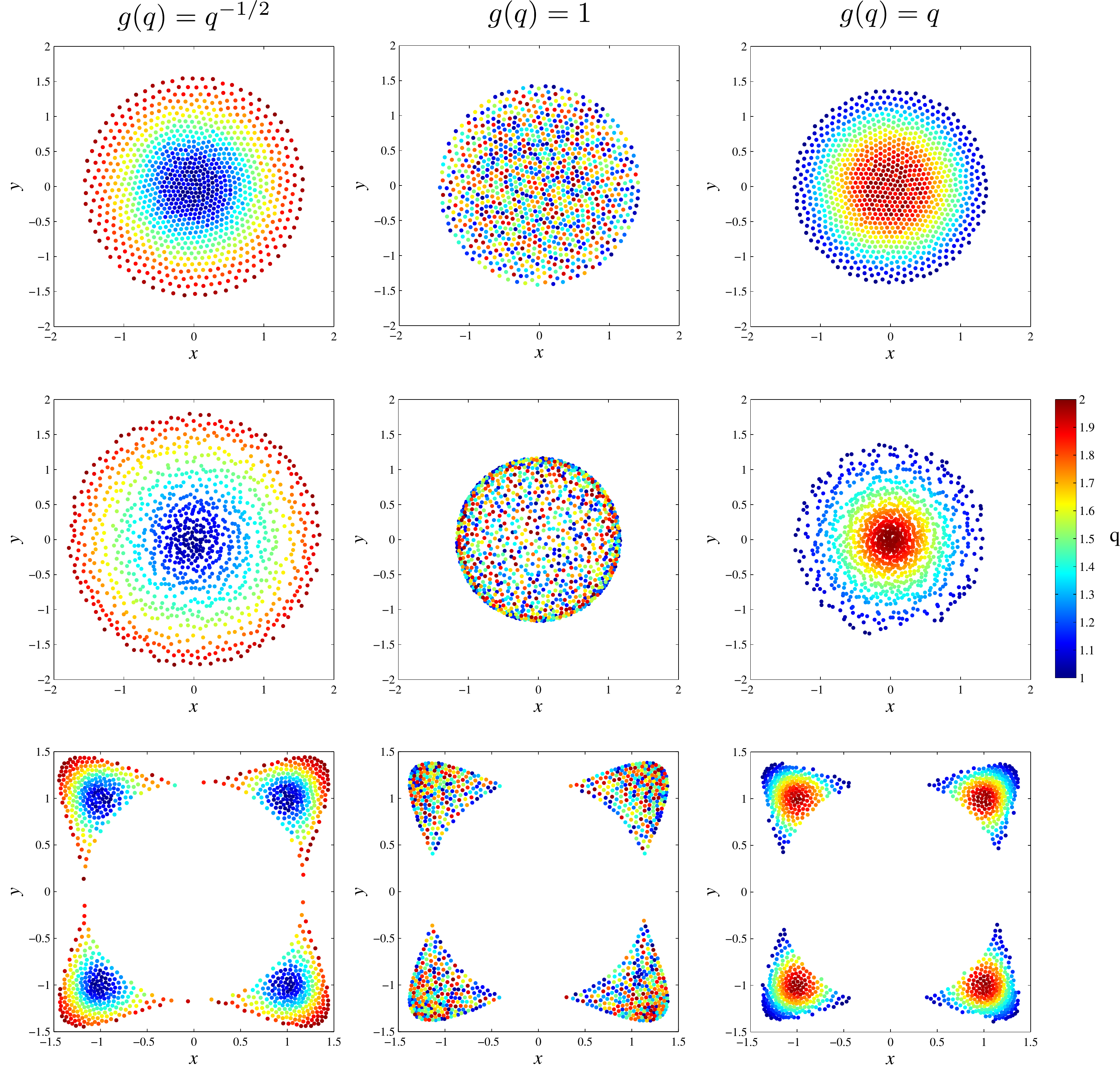}
	\caption{The ordering transition in general gases. Each column corresponds to one different choice of $g$ (decreasing, constant and increasing). Top two lines correspond to super-Coulomb ($\eta=0.5$) and sub-Coulomb ($\eta=-0.5$) interactions with a quadratic potential $V(z)=\vert z\vert^2$. Bottom line corresponds to Coulomb interactions ($\eta=0$) and non-quadratic potential $V(z)=\frac 1 2 (x^4+y^4) - \vert z\vert^2$. In all cases $N=1000$, $\nu(q)$ uniform in $Q=[1,2]$. }
	\label{fig:generalgases}
\end{figure}

\section{General heterogeneous gases on manifolds}\label{append:manifolds}
We discuss here the distributions of heterogeneous gases on smooth manifolds $\Gamma$ of dimension $k \geq 2$ in $\R^d$. The homogeneous case was analyzed in~\cite{Saff:97,hardin2005}. In these cases, the ordering transition also occurs, and can be proved exactly in the same fashion as done for general gases (proposition~\ref{pro:GeneralGasesMacro}). In detail, denoting $(e_1(x),\cdots, e_k(x))$ a basis of the tangent space of the manifold at $x\in \Gamma$, a particle of charge $q$ is at equilibrium at location $x$ if and only if:
\[\left( g(q) \,\nabla V(x) - 2\int q'\nabla_{x} W(|x-y|)d\mu(q',y)\right)\cdot e_p(x) =0 \qquad \forall p\in\{1,\cdots,k\}\]
\begin{enumerate}
	\item for $g$ constant, it is clear that this condition is independent of $q$ and therefore several configurations are at equilibrium, including disordered distributions.
	\item if for all $i\in \{1,\cdots,k\}$ and all $x\in \Gamma$, $\nabla V(x) \cdot e_i(x)\neq 0$, then at a specific location $x$, only particles of a specific charge are at equilibrium. Depending on the specific choice of manifold and potential, we may therefore obtain ordered configurations depending on the monotonicity of $g$. 
\end{enumerate}
Three examples of gases on two-dimensional manifolds in $\R^3$ are presented in Fig.~\ref{fig:manifolds}. The first two examples correspond to gases on the sphere, with logarithmic interactions and distinct external potentials, and the last example consists of a gas on the two-dimensional torus with logarithmic interactions. All exhibit the ordering transition. 
\begin{figure}[htbp]
	\centering
		\includegraphics[width=.85\textwidth]{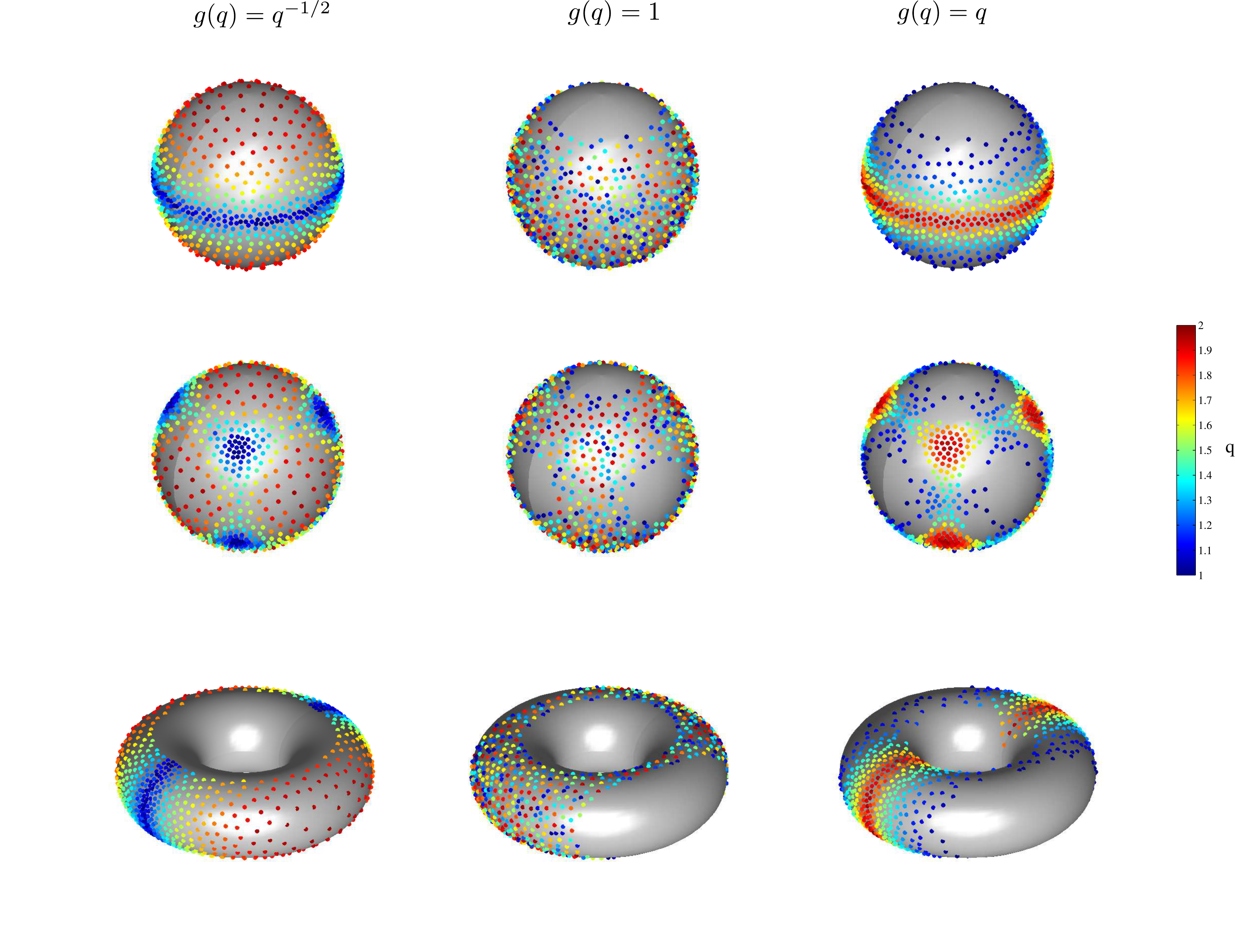}
	\caption{Ordering transition in heterogeneous gases on two-dimensional manifolds. Each column corresponds to one different choice of $g$ (decreasing, constant and increasing). First and second row correspond to gases on the unit sphere $\S_2$ with a potential $V(x,y,z)=z^2$ (first row) or $V(x,y,z)=\frac 1 2 (x^4+y^4+z^4) - (x^2+y^2+z^2)$ (second row). The third row corresponds to a gas on the torus $\left(1-\sqrt{x^2+y^2}\right)^2+z^2=\left(\frac12\right)^2 $ with  potential $V(x,y,z)=y^2$. In all cases $N=1000$, $\nu(q)$ uniform in $Q=[1,2]$. }
	\label{fig:manifolds}
\end{figure}

\bibliographystyle{plain}
\bibliography{biblio}

\end{document}